\documentclass[journal,10pt,letterpaper,twocolumn]{IEEEtran}

\usepackage{macros}
\usepackage{amsmath}
\usepackage{amssymb}
\usepackage{amsthm}
\usepackage{etoolbox}
\usepackage{color,colortbl}
\usepackage{mathtools}
\usepackage{enumerate}
\usepackage[normalem]{ulem}
\usepackage{pgf}
\usepackage{tikz}
\usetikzlibrary{trees,decorations,arrows,automata,shadows,positioning,plotmarks,backgrounds,shapes,shapes.misc}
\usetikzlibrary{calc,matrix,fit,petri,decorations.pathmorphing,patterns}
\usetikzlibrary{decorations.pathreplacing,decorations.markings,shapes.geometric,calc}
\usepackage{paralist}
\usepackage{stmaryrd}
\usepackage{xspace}
\usepackage{subfigure}
\usepackage{float}
\usepackage{csquotes}

\newtheorem{proposition}{Proposition}

\newtheorem{definition}{Definition}
\newtheorem{lemma}{Lemma}
\newtheorem{assumption}{Assumption}
\newtheorem{theorem}{Theorem}

\newtheorem{corollary}{Corollary}
\theoremstyle{definition}
\newtheorem{remark}{Remark}
\newtheorem{example}{Example}

\newcommand{\SEZS}[1]{\textcolor{violet}{#1}}

\newcommand{\set}[1]{{\{ #1 \}}}

\newcommand{\Sys}[1]{\ifstrempty{#1}{S}{S_{#1}}}
\newcommand{\real}[1]{\ifstrempty{#1}{\mathbb{R}}{\mathbb{R}^{#1}}}
\newcommand{\posreal}[1]{\real{#1}_{>0}}
\newcommand{\poszreal}{{\real{}_{\geq 0}}}

\newcommand{\delimit}{\hspace{1.5mm}}

\newcommand{\norm}[1]{\|\! ~#1~\!\|}

\newcommand{\dist}[2]{d(#1,#2)}

\newcommand{\DiscSpace}[2]{[#1]_{#2}}

\newcommand{\ConjFunc}[2]{#1~\circ ~#2}

\newcommand{\KInf}{\mathcal{K}_\infty}

\newcommand{\dISS}{\delta\text{-ISS}}

\newcommand{\kifunc}{\mathcal{K}_{\infty}}
\newcommand{\st}{s.t.\xspace}

\newcommand{\transpose}[1]{[#1]^T}
\newcommand{\defeq}{:=}

\newcommand{\DQnt}[1]{\ifstrempty{#1}{\tilde{\varepsilon}}{\tilde{\varepsilon}_{#1}}}
\newcommand{\vmetric}{\mathbf{e}}

\newcommand{\Def}{Def.\xspace}
\newcommand{\Thm}{Thm.\xspace}
\newcommand{\Lem}{Lemma\xspace}
\newcommand{\Sec}{Sec.\xspace}
\newcommand{\Page}{p.\xspace}
\newcommand{\Eqn}{Eqn.\xspace}
\newcommand{\Ass}{Assump.\xspace}
\newcommand{\Inq}{Inq.\xspace}

\newcommand{\Fig}{Fig.\xspace}

\newcommand{\Prop}{Prop.\xspace}
\newcommand{\Cor}{Cor.\xspace}

\newcommand{\partderiv}[2]{\dfrac{\partial #1}{\partial #2}}

\newcommand{\Lypv}[1]{\ifstrempty{#1}{V}{V_{#1}}}
\newcommand{\gammafunc}[1]{\gamma_{#1}}

\newcommand{\alphalow}[1]{\ifstrempty{#1}{\underline{\alpha}}{\underline{\alpha}_{#1}}}
\newcommand{\alphahigh}[1]{\ifstrempty{#1}{\overline{\alpha}}{\overline{\alpha}_{#1}}}
\newcommand{\lambdaL}[1]{\ifstrempty{#1}{\lambda}{\lambda_{#1}}}

\newcommand{\sigmad}[1]{\ifstrempty{#1}{\sigma_{d}}{\sigma_{d,#1}}}
\newcommand{\sigmau}[1]{\ifstrempty{#1}{\sigma_{u}}{\sigma_{u,#1}}}
\newcommand{\eqclass}[1]{\mathit{R}_{\varepsilon_{#1}\tilde{\varepsilon}_{#1}}}

\newcommand{\paraeps}[1]{(\varepsilon_{#1},\tilde{\varepsilon}_{#1})}

\newcommand{\StateSpace}[1]{\ifstrempty{#1}{\mathbb{R}^n}{\mathbb{R}^{n_{#1}}}}

\newcommand{\IPECurve}[1]{\ifstrempty{#1}{\mathcal{U}}{\mathcal{U}_{#1}}}
\newcommand{\IPICurve}[1]{\ifstrempty{#1}{\mathcal{W}}{\mathcal{W}_{#1}}}
\newcommand{\TransFunc}[1]{\ifstrempty{#1}{f}{f_{#1}}}

\newcommand{\ContSysSymb}[1]{\ifstrempty{#1}{\Sigma}{\Sigma_{#1}}}

\newcommand{\Wt}{\widetilde{W}}
\newcommand{\wt}{\tilde{w}}

\newcommand{\IntCon}{\mathcal{I}}
\newcommand{\nbr}{\mathcal{N}}

\newcommand{\State}[1]{\ifstrempty{#1}{X}{X_{#1}}}
\newcommand{\state}[1]{\ifstrempty{#1}{x}{x_{#1}}}

\newcommand{\IPE}[1]{\ifstrempty{#1}{U}{U_{#1}}}
\newcommand{\ipe}[1]{\ifstrempty{#1}{u}{u_{#1}}}
\newcommand{\IPI}[1]{\ifstrempty{#1}{W}{W_{#1}}}
\newcommand{\ipi}[1]{\ifstrempty{#1}{w}{w_{#1}}}

\newcommand{\TimeQnt}[1]{\ifstrempty{#1}{\tau}{\tau_{#1}}}
\newcommand{\ContSysDT}[1]{\mathcal{P}_{\TimeQnt{}}(#1)}

\newcommand{\StateDT}[1]{\ifstrempty{#1}{X_{\TimeQnt{}}}{X_{#1,\TimeQnt{}}}}

\newcommand{\IPECurveDT}[1]{\ifstrempty{#1}{\mathcal{U}_\TimeQnt{}}{\mathcal{U}_{#1,\TimeQnt{}}}}
\newcommand{\IPICurveDT}[1]{\ifstrempty{#1}{\mathcal{W}_\TimeQnt{}}{\mathcal{W}_{#1,\TimeQnt{}}}}
\newcommand{\TransFuncDT}[2]{\ifstrempty{#1}{\ifstrempty{#2}{\xrightarrow[\TimeQnt{}]{}}{\xrightarrow[\TimeQnt{}]{#2}}}{\ifstrempty{#2}{\xrightarrow[\TimeQnt{}]{}}{\xrightarrow[\TimeQnt{}]{#2}}}}

\newcommand{\StateTraj}[2]{\ifstrempty{#1}{\ifstrempty{#2}{\xi}{\xi_{#2}}}{\ifstrempty{#2}{\xi_{#1}}{\xi_{#1,#2}}}} % #1: system index, #2: initial state and input
\newcommand{\StateTrajDot}[1]{\ifstrempty{#1}{\dot{\xi}}{\dot{\xi}_{#1}}}
\newcommand{\IPETraj}[1]{\ifstrempty{#1}{\mu}{\mu_{#1}}}
\newcommand{\IPITraj}[1]{\ifstrempty{#1}{\nu}{\nu_{#1}}}

\newcommand{\StateQnt}[1]{\ifstrempty{#1}{\eta}{\eta_{#1}}}
\newcommand{\IPQnt}[1]{\ifstrempty{#1}{\omega}{\omega_{#1}}}

\newcommand{\ContSysA}[2]{\ifstrempty{#1}{\mathcal{P}_{\AllQnt{}}(#2)}{\mathcal{P}_{\AllQnt{#1}}(#2)}}
\newcommand{\ContSysAi}[1]{\mathcal{P}_{\AllQnt{i}}(#1)}
\newcommand{\ContSysAc}[2]{\ifstrempty{#1}{\mathcal{P}^c_{\AllQnt{}}(#2)}{\mathcal{P}^c_{\AllQnt{#1}}(#2)}}

\newcommand{\AllQnt}[1]{\ifstrempty{#1}{{\TimeQnt{}\StateQnt{}\IPQnt{}}}{{\TimeQnt{}\StateQnt{#1}\IPQnt{#1}}}}

\newcommand{\StateA}[1]{\ifstrempty{#1}{\State{\AllQnt{}}}{\State{#1,\AllQnt{#1}}}}

\newcommand{\IPECurveA}[1]{\ifstrempty{#1}{\IPECurve{\AllQnt{}}}{\IPECurve{#1,\AllQnt{#1}}}}

\newcommand{\IPICurveA}[1]{\ifstrempty{#1}{\IPICurve{\AllQnt{}}}{\IPICurve{#1,\AllQnt{#1}}}}

\newcommand{\deltaA}[1]{\ifstrempty{#1}{\delta_{\AllQnt{}}}{\delta_{#1,\AllQnt{#1}}}}
\newcommand{\TransFuncA}[2]{\ifstrempty{#1}{\ifstrempty{#2}{\xrightarrow[\AllQnt{}]{}}{\xrightarrow[\AllQnt{}]{#2}}}{\ifstrempty{#2}{\xrightarrow[#1,\AllQnt{}]{}}{\xrightarrow[#1,\AllQnt{}]{#2}}}}

\newcommand{\AllQntb}[1]{\ifstrempty{#1}{\TimeQnt{}\StateQnt{}\IPQnt{}}{\TimeQnt{}\StateQnt{#1}\IPQnt{#1}}}

\newcommand{\StateAb}[1]{\ifstrempty{#1}{\State{\AllQntb{}}}{\State{#1,\AllQntb{#1}}}}

\begin{document}
% \title{Compositional Synthesis of Finite State Abstractions }
\title{Compositional Abstraction-Based Controller Synthesis for Continuous-Time Systems}

\author{Kaushik Mallik, Anne-Kathrin Schmuck, Sadegh Soudjani, Rupak Majumdar% <-this % stops a space
% \thanks{This work was not supported by any organization}% <-this % stops a space
\thanks{All authors are with MPI-SWS, Kaiserslautern, Germany. {\tt\small \{kmallik,akschmuck,sadegh,rupak\}@mpi-sws.org}}% <-this % stops a space
% \thanks{This work was partially supported by the HYCON2 Network of excellence (FP7 ICT 257462).}% <-this % stops a space
}

% make the title area
\maketitle

\begin{abstract}
Controller synthesis techniques for continuous systems with respect to temporal logic specifications typically use a finite-state symbolic abstraction of the system model. 
Constructing this abstraction for the entire system is computationally expensive, and does not
exploit natural decompositions of many systems into interacting components. 
We describe a methodology for compositional symbolic abstraction 
to help scale controller synthesis for temporal logic to larger systems.

We introduce a new relation, called (approximate)
\emph{disturbance bisimulation}, as the basis for compositional symbolic abstractions.
Disturbance bisimulation strengthens the standard approximate alternating bisimulation relation
used in control.
It extends naturally to systems which are composed of weakly interconnected sub-components possibly connected 
in feedback, and models the coupling signals as disturbances.
After proving this composability of disturbance bisimulation for metric systems we apply this result to the compositional abstraction of networks of input-to-state stable deterministic non-linear control systems. 
We give conditions that allow to construct finite-state abstractions compositionally for each component in such a network, so that 
the abstractions are simultaneously disturbance bisimilar to their continuous counterparts. Combining these two results, we show conditions under which one can compositionally abstract a network of non-linear control systems in a modular way while ensuring that the final composed abstraction is disturbance bisimilar to the original system.
% The second, called \emph{composition of approximation}, shows a disturbance
% bisimulation relationship between compositions of the finite symbolic models and
% the composition of the original network, when a small-gain condition is satisfied.

% Our results are non-trivial because each component's dynamics is influenced by 
% the states of other interacting components, whose values
% are only approximately known.
We discuss how we get a compositional abstraction-based controller synthesis
methodology for networks of such systems against 
local temporal specifications as a by-product of our construction. 
\end{abstract}

% % Note that keywords are not normally used for peerreview papers.
% \begin{IEEEkeywords}
% 
% \end{IEEEkeywords}

\section{Introduction}
\label{sec:intro}

% A symbolic model for a continuous dynamical system is a discrete approximation of the continuous system
% such that a controller designed on this abstraction can be refined to a controller for the
% original system.
Symbolic models for continuous dynamical systems enable powerful automata-theoretic techniques for controller design
for $\omega$-regular specifications to be applied to continuous systems.
In this methodology, one starts with a continuous dynamical system and an approximation factor $\varepsilon$, 
and constructs a finite-state abstraction whose trajectories
are guaranteed to be within a distance of $\varepsilon$ to the original system and vice versa
\cite{TabuadaBook,girard1,PolaGT08,pola2009symbolic,ZamaniPolaMazoTabuada_2012}.
The approximation is usually formalized using 
$\varepsilon$-approximate alternating bisimulation relations, 
which has the property that a controller synthesized for the abstraction can be automatically refined into controller for the original system. 
Under the assumption of incremental input-to-state stability, one can algorithmically construct a finite-state discrete system which is $\varepsilon$-approximately alternatingly bisimilar to the original continuous system.
Since one can also algorithmically synthesize controllers for $\omega$-regular properties for 
discrete systems (see, e.g., \cite{EmersonJutla91,MPS95}), 
this provides an automatic controller synthesis technique for continuous systems.
The methodology is integrated into controller synthesis tools \cite{Pessoa,Scots}, and
has been recently applied to large case studies in adaptive cruise control \cite{NilssonHBCAGOPT16} and bipedal robots \cite{AmesTSMKRG15}.
It has also been extended to systems with disturbances \cite{pola2009symbolic,BorriPB12} or to stochastic systems \cite{ZamaniEMAL14,ZamaniAG15,ZamaniRE15}.%, and, very recently, to networked control systems \cite{BorriPB14,ZamaniMA14}.

The computational bottleneck of this approach is the expensive abstraction step (typically exponential in the dimension) which limits its applicability to real systems. However, in practice, many systems are designed using interacting networks of smaller dynamically coupled
components. 
One would imagine that each component can be abstracted separately, by modeling the states of the neighboring components influencing its dynamics as disturbance signals. %in an environment consisting,abstractly, of the other components. 

Performing controller synthesis on these separate component abstractions locally, results in a decentralized control architecture where each component is connected to its individual controller and controllers of different components do not communicate.
Such a decentralized control architecture must treat neighboring components as adversaries. Thus locally synthesized controllers have to be able to counteract all possible disturbances coming from neighboring components. Therefore, as well known in classical control theory, this architecture only results in satisfying controller performance if couplings between dynamics of the interconnected components are small (See e.g. \cite[Chap.~21]{goodwin2001control}). 

In this paper we show how decentralized controllers for a network of weakly coupled \emph{nonlinear continuous-time} dynamical systems can be synthesized via the abstract controller synthesis paradigm discussed before. The main ingredient of our approach is a compositional abstraction technique that allows us to apply the standard controller synthesis for each local abstraction.

Compositional abstractions for networked components are challenging due to the following observation. If we apply the usual approach to construct finite-state abstractions (using $\varepsilon$-approximate alternating bisimulation relations as in \cite{PolaGT08,pola2009symbolic,ZamaniPolaMazoTabuada_2012}) to an individual component in the network, % using the notion of approximate alternating bisimulations \cite{PolaGT08,pola2009symbolic,ZamaniPolaMazoTabuada_2012}, 
its abstraction is defined over a discretized version of the component's state space. By treating state trajectories of neighboring components as disturbance inputs, discretizing the state space of one component also discretizes (parts of) the disturbance space of it's neighboring components. This gives rise to a mismatch of the disturbance signals of each component and it's abstraction (as the former is continuous while the latter is piecewise constant). This mismatch is bounded by the abstraction parameter $\varepsilon_i$ but only at sampling time instances. 
We therefore need to reason about the similarity of two systems (the component and the abstraction in this case) whose disturbance trajectories are different and whose mismatch might increase during inter-sampling periods. 

To deal with this challenge we introduce a new binary relation, called \emph{disturbance bisimulation} with two approximation parameters 
$\paraeps{}$ and provide conditions for the class of nonlinear continuous-time control systems that bound the error during inter-sampling periods to allow for the construction of disturbance bisimilar abstractions. %This leads to the following contributions of the paper. 

\subsection*{Outline and Contributions}
This paper consists of three parts: Part I (Sec.~\ref{sec:MetricSystems}-\ref{sec:DisturbanceBisimulation}),  Part II (Sec.~\ref{sec:ControlSystems}-\ref{sec:composition}), 
and Part III (\Sec~\ref{sec:control}-\ref{sec:Example}).

\emph{Part I} focuses on \emph{metric systems} as defined in Sec.~\ref{sec:MetricSystems} and introduces disturbance bisimulation for this system class in Sec.~\ref{sec:DisturbanceBisimulation}. As our first contribution, we show that disturbance bisimulation naturally extends to networks of metric systems.
% \begin{compactenum}[1)]
%  \item[(C1)] For two sets of equally interconnected metric systems $\set{S_{i}}_{i\in I}$ and  $\set{\hat{S}_{i}}_{i\in I}$, s.t. for all $i\in I$, $S_{i}$ and $\hat{S}_{i}$ are disturbance bisimilar, there exists a simple condition on the parameters $\paraeps{i}$ such that their composition is also disturbance bisimilar.
% \end{compactenum}
% While it was shown in \cite{majumdar2016compositional} that there is a way to construct a disturbance bisimilar abstraction of a metric system generated by time-sampling a general, non-linear control system, we will shown in the remainder of this paper that this is also true for the class for \emph{stochastic control systems}. %By this we will provide a technique to compositionally abstract stochastic control systems for controller synthesis. 

\emph{Part II} applies the compositional abstraction result for metric systems from Part I to the class of \emph{input-to-state stable deterministic non-linear control systems}, defined in \Sec~\ref{sec:ControlSystems}. First, we focus on a single control system $\Sigma$ in \Sec~\ref{sec:MonolithicAbstraction} which has the additional property that the growth-rate of its disturbance is bounded during the inter-sampling period. As our second contribution we show how to construct a finite state symbolic abstraction $\hat{\Sigma}$ (which is a metric system) of $\Sigma$ s.t. $\hat{\Sigma}$ is disturbance bisimilar to the sampled time model (which is again a metric system) of $\Sigma$. As our third contribution we given conditions under which this result can be combined with the one from Part I to provide a compositional abstraction method for networks of control systems in \Sec~\ref{sec:composition}. Intuitively, the obtained conditions limit the allowed coupling between neighboring subsystems and link the 
abstraction parameter of the state space of one component with the parameters bounding the disturbance mismatch of its neighboring components.

\emph{Part III} discusses a decentralized methodology for controller synthesis in networked systems based on disturbance bisimulations (\Sec~\ref{sec:control}).
% As our notion of disturbance bisimulation strengthens the notion of $\varepsilon$-approximate alternating bisimulation, our approach leads naturally to a decentralized methodology for controller synthesis in networked systems, which is briefly discussed in \Sec~\ref{sec:control}. 
To show the strength of our approach, we apply our decentralized abstraction-based controller synthesis method to 
a system consisting of 200 components and a total of 400 state variables in \Sec~\ref{sec:Example}.

\subsection*{Related Work}
%Very recently some efforts have been made in this direction.
% 
Conceptually the closest related works are \cite{tazaki2008bisimilar,rungger2015compositional,zamani2016compositional}.
In \cite{tazaki2008bisimilar}, the authors presented a compositional approach for finite state abstractions of a network of control systems. Their interconnection-compatible approximate bisimulation is similar to our disturbance bisimulation. However, their approach is only applicable to \emph{discrete-time linear} systems. 
In \cite{rungger2015compositional}, the authors presented a compositional approach to construct approximate abstractions which perform a model order reduction
from one continuous system to another \emph{continuous} system with fewer state variables. %Hence, their approach could be used as a preprocessing step for a compositional symbolic abstraction.
% However, techniques developed in their context, particularly a composition lemma using a small-gain assumption, have an analogue in our setting. 
In \cite{zamani2016compositional}, a similar approach as ours was presented for solving a \emph{continuous} compositional abstraction synthesis problem using 
ideas from dissipativity theory; their joint storage functions use the same quantifier alternation as our disturbance bisimulation. %\KM{How is this work different?}

Pola et al.~\cite{PolaPB14,PolaPBTAC16} proposed a compositional abstraction technique for networked continuous systems
based on approximate bisimulation. Unfortunately, the use of bisimulation introduces the 
unrealistic assumption that components are free to choose the state trajectories of their neighboring components (recall that
a bisimulation relation is allowed to \emph{pick} a suitable matching trajectory).
This is not realistic in a compositional setting, in which one component does not control the trajectory of other components
in the system. 

Dallal et al.~\cite{DallalTabuada2015} proposed a compositional controller synthesis algorithm for discrete-time systems based on a small-gain-theorem 
and assume-guarantee techniques. Here, state variables of neighboring components are over-approximated by sets, and local abstractions are computed under this additional source for non-determinism. This provides a different way to incorporate disturbances caused by neighboring components into the abstraction of a local components. In contrast to our work only discrete-time systems and persistence specifications are treated.

Most works on abstraction based controller synthesis only give guarantees on the closeness of trajectories at sampling instances or discuss only the abstraction of discrete time systems. Notable exceptions are \cite{LiuOzay_2014} and \cite{Girard2016_InterSampling}, where the robustness-margins introduced in \cite{LiuOzay_2014} have a similar effect as the growth bound introduced in our work.

% They only treat persistence specifications 
% and no notion of $\varepsilon$-closedness is employed.
% Hence synthesized abstract controller may not be refined to the original continuous system.
%A similar approach was recently taken in \cite{zamani2016compositional} for solving a continuous compositional abstraction synthesis problem using 
%ideas from dissipativity theory; their joint storage functions use the same quantifier alternation as our disturbance bisimulation.

\section{Metric Systems}\label{sec:MetricSystems}

This section introduces metric systems and networks of such systems as the underlying system models used in this paper.

% \subsection{Metric Systems} \label{sec:monolithic abstraction}

% % We now define the abstract semantics of control systems based on concurrent games \cite{AHKV98}. 
% In this paper we restrict our attention to norm-induced metric systems that are defined over an euclidean vector space and are time sampled w.r.t.\ a time sampling parameter\footnote{We only use one single time sampling parameter called $\tau$ everywhere in this paper and therefore usually only implicitly assume that it is given.} $\tau \in \posreal{}$.
% We point out that our results can be readily extended to general metric systems at the cost of more complex notation.
% % This is for notational convinience only. All results can be extended to arbitrary metric systems in a streight forward manner.

\subsection{Preliminaries}

We use the symbols $\mathbb{N}$, $\real{}$, $\posreal{}$, $\poszreal$ and $\mathbb{Z}$ to denote the set of natural, 
real, positive real, nonnegative real numbers and integers, respectively. 
The symbols $I_n$, $0_n$, and $0_{n\times{m}}$ denote the identity matrix, the zero vector and the zero matrix 
in $\real{n\times{n}}$, $\real{n}$, and $\real{n\times{m}}$, respectively. Given a vector \mbox{$x\in\real{n}$}, we denote by $x_{i}$ the $i$--th element of $x$ and
by $\norm{x}$ the infinity norm of $x$.

% \begin{definition}\label{def:metricsystem}
Given a time sampling parameter $\tau\in \posreal{}$, a \emph{metric system}\footnote{
Often, metric systems are defined with an additional output space and an output map from states to the output space.
We omit the output space for notational simplicity; for us, the state
and the output space coincide, and the output map is the identity function.}
$S=(X,U,\mathcal{U}_\tau,W,\mathcal{W}_\tau,\delta_\tau)$
consists of
a (possibly infinite) set of states $X\subseteq \real{n}$ equipped with a metric $d:X\times X \rightarrow \poszreal$, %\footnote{We only work with euclidean vector space and norm induced metric in this paper, and point out that our results can be readily extended to general metric spaces at the cost of more complex notation.}
a set of piecewise constant inputs $\mathcal{U}_\tau$ of duration $\tau$ taking values in $U\subseteq \real{m}$, i.e.,
\begin{subequations}\label{equ:def:mathcalUW}
 \begin{align}\label{equ:def:mathcalU}
	\mathcal{U}_\tau & = \set{
		 \mu:[0,\TimeQnt{}]\rightarrow U \mid
		 \forall t_1,t_2\in[0,\TimeQnt{}]~.~\mu(t_1) = \mu(t_2)},
\end{align}
%a set of piecewise constant disturbances $\mathcal{W}_\tau$ taking values in the set $W\subseteq \real{p}$, i.e.,
a set of disturbances $\mathcal{W}_\tau$ taking values in $W\subseteq \real{p}$, i.e.,
%\begin{align*}
%	\mathcal{W}_\tau & = \set{
%		\nu:[0,\TimeQnt{}]\rightarrow W \mid
%		 \forall t_1,t_2\in[0,\TimeQnt{}]~.~\nu(t_1) = \nu(t_2)},
%\end{align*}
\begin{align}\label{equ:def:mathcalW}
	\mathcal{W}_\tau & \subseteq \set{
		\nu:[0,\TimeQnt{}]\rightarrow W},
\end{align}
%and a transition function $\delta: X\times \mathcal{U}_\tau \times \mathcal{W}_\tau \rightarrow X$. For any given $x,x'\in X$, $\mu\in \IPECurveDT{}$ and $\nu\in \IPICurveDT{}$ s.t. $\delta(x,\mu,\nu) = x'$, we write $x \xrightarrow[\tau]{\mu,\nu} x'$.}
and a transition function $\delta_\tau: X\times \mathcal{U}_\tau \times \mathcal{W}_\tau \rightarrow 2^{X}$. We write $x\xrightarrow[\tau]{\mu, \nu} x'$ when $x'\in\delta_\tau(x,\mu, \nu)$, and we denote the unique value of $\mu\in\mathcal{U}$ over $[0,\tau]$ by  $u_\mu\in U$.

% Let $\mathcal{X}_\tau$ be the set of state trajectories over $[0,\TimeQnt{}]$, i.e.,
% \begin{align}\label{equ:def:mathcalX}
% 	\mathcal{X}_\tau & \subseteq \set{
% 		\xi:[0,\TimeQnt{}]\rightarrow X}.
% \end{align}
\end{subequations}
% Then we extend $\delta$ to state trajectories by defining $\xi\in\delta(\xi(0),\mu,\nu)$ if for all $t\in [0,\tau]$, $\xi(t)\in\delta(t,\xi(0),\mu,\nu)$ holds.
% 

%\new{If the metric system $S$ is undisturbed, we define $W=\mathcal{W}_\tau=\emptyset$. In this case we occasionally represent $S$ by the tuple $S = (X, U, \mathcal{U}_\tau, \delta_\tau)$ and use $\delta_\tau:X\times \mathcal{U}_\tau\rightarrow 2^X$.}
%\KM{I prefer this version over having $W=\set{0}$ (for linear systems it does not matter). Consider the metric system with $\delta_\tau(x,\mu,\nu) = x*u_\mu*\nu(0)$. Clearly simply setting $\nu$ to be the constant signal taking value $0$ all the time does not make the system unperturbed. The unperturbed system is rather $\delta_\tau(x,\mu,\nu) = x*u_\mu$. If you agree, please copy and paste this (with the replacements of $\delta_\tau$ by $f$) by replacing the third paragraph of Section IV.}

If the metric system $S$ is undisturbed, we define $W=\set{0}$. In this case we occasionally represent $S$ by the tuple $S = (X, U, \mathcal{U}_\tau, \delta_\tau)$ and use $\delta_\tau:X\times \mathcal{U}_\tau\rightarrow 2^X$ with the understanding that $x'\in\delta_\tau(x,\mu, \nu)$ holds for the zero trajectory $\IPITraj{}:{\poszreal}\rightarrow \set{0}$ whenever $x'\in\delta_\tau(x,\mu)$.
% \KM{This is necessary as this situation appears in Example 1.}

By slightly abusing notation we write $x'=\delta_{\tau}(x,\mu, \nu)$ as a short form when the set $\delta_{\tau}(x,\mu, \nu) = \set{x'}$ is singleton.
% We call the metric system $S$ deterministic if for all $x\in X$, $\mu\in\mathcal{U}_\tau$ and $\nu\in\mathcal{W}_\tau$ we have $|\delta_\tau(x,\mu, \nu)|\leq 1$. For deterministic metric systems we write $x'=\delta_{\tau}(x,\mu, \nu)$ as the short form of $\set{x'}=\delta_\tau(x,\mu, \nu)$.% This notation is extended to trajectories in the obvious way.
% \SEZS{Sadegh: Non-blocking? Do we need deterministic metric systems or just the definition with single element?}

If $X$, $\mathcal{U}_{\tau}$ and $\mathcal{W}_{\tau}$ are finite (resp. countable), $S$ is called \emph{finite} (resp. \emph{countable}). 
% \new{In particular, we will construct finite metric systems as symbolic abstractions of control systems in this paper, which will contain only piecewise constatant state, input and disturbance curves taking values from finite sets $X$, $U$ and $W$. We call such metric systems, \emph{abstract} metric systems. }
We also assign to a transition $x'=\delta_{\tau}(x,\mu, \nu)$  any continuous time evolution $\xi:[0,\tau]\rightarrow X$ \st $\xi(0) = x$ and $\xi(\tau) = x'$.

% \end{definition}

\subsection{Networks of Metric Systems}\label{sec:Comp:Metric}

First let us introduce some notation. 
% In this section we briefly recall the definition of a \emph{network of systems} from \cite{majumdar2016compositional} and apply it to stochastic systems.
% 
% introduce a network of interconnected stochastic systems by allowing states of one system to be 
% fed back to other systems, which are treated as disturbances. We furthermore define how to generate composed systems from such networks. To do so we first define a network of systems as follows.
% 
Let $I$ be an index set (e.g., $I = \set{1,\ldots, N}$ for some natural number $N$) and let
$\IntCon\subseteq I\times I$ be a binary irreflexive \emph{connectivity relation} on $I$. 
Furthermore, let $I'\subseteq I$ be a subset of systems with $\IntCon' := (I'\times I')\cap \IntCon$.
For $i\in I$ we define %\footnote{We omit the suffix and write $\nbr(\cdot)$, whenever the interconnection relation $\mathcal{I}$ is clear from the context.} 
$\nbr_{\mathcal{I}}(i) = \set{j\mid (j,i)\in \IntCon}$ and extend this notion to subsets of systems $I'\subseteq I$ as $\nbr_{\mathcal{I}}(I') = \set{j\mid \exists i\in I'.j\in \nbr_{\mathcal{I}\setminus\mathcal{I}'}(i)}$. %\new{\neq\emptyset}$. 
Intuitively, a set of systems can be imagined to be the set of vertices $\set{1,2,\ldots,|I|}$ of a directed graph $\mathcal{G}$, and $\mathcal{I}$ to be the corresponding adjacency relation. Given any vertex $i$ of $\mathcal{G}$, the set of incoming (resp. outgoing) edges are the inputs (resp. outputs) of a subsystem $i$, and $\nbr_{\mathcal{I}}(i)$ is the set of neighboring vertices from which the incoming edges originate.

Let $S_i = (X_i, U_i, \mathcal{U}_{\tau,i}, W_i, \mathcal{W}_{\tau,i}, \delta_{\tau,i})$,
for $i\in I$, be a metric system with metric $d_i$. Then we say that $\set{S_i}_{i\in I}$ are \emph{compatible for composition} w.r.t.\ the interconnection relation $\mathcal{I}$, if 
for each $i\in I$, we have $W_i = \prod_{j\in \nbr_\mathcal{I}(i)} {X_j}$, i.e., the disturbance input space of $S_i$ is the same as the Cartesian product of the state spaces of all the neighbors in $\nbr_{\mathcal{I}}(i)$.
 %
 %\SEZS{Thus any $w_i\in W_i$ can be written as $w_i=\prod_{j\in \nbr_\mathcal{I}(i)} \set{x_j}$ for $x_j\in X_j$.}
 %
% \AKS{I think we need to have the previous version:}
By slightly abusing notation we write $w_i=\prod_{j\in \nbr_\mathcal{I}(i)} \set{x_j}$ for $x_j\in X_j$ and $w_i\in W_i$ as a short form for the single element of the set $\prod_{j\in \nbr_\mathcal{I}(i)} \set{x_j}$. We extend this notation to all sets with a single element.
 
As $I'$ is a subset of all systems in the network, we divide the set of disturbances $W_i$ for any $i\in I'$ into the sets of coupling and external disturbances, defined by
$W_i^c = \prod_{j\in \nbr_{\mathcal{I}'}(i)} {X_j}$ and $W_i^e = \prod_{j\in\nbr_{\IntCon\setminus\IntCon'}(i)} {X_j}$, respectively. 

We extend the metrics $d_j$ on $X_j$, $j\in \nbr_{\mathcal{I}}(i)$, to the vector valued metric $\vmetric:W_i\times W_i\rightarrow\real{|\nbr_{\mathcal{I}}(i)|}_{\geq 0}$ on $W_i$ \st for any $w_i=\prod_{j\in \nbr_\mathcal{I}(i)} \set{x_j}\in W_i$ and $w_i'=\prod_{j\in \nbr_\mathcal{I}(i)} \set{x_j'}\in W_i$,
\begin{equation}\label{equ:defvmetric_composed}
 \vmetric(w_i,w_{i}'):=\hspace{-0.3cm}\prod_{j\in \nbr_{\mathcal{I}}(i)} \lbrace d_j(x_j,x_j') \rbrace.
 								%=\hspace{-0.3cm}\prod_{j\in \nbr_{\mathcal{I}}(i)} \lbrace \norm{x_j-x_j'} \rbrace.
\end{equation}
%where $\nu_j=\nu|_{X_{j}}$.
Intuitively, $\vmetric(w_i,w_i')$ is a vector with dimension $|\nbr_{\mathcal{I}}(i)|$, where the $j$-th entry measures the mismatch of the respective state vector of the $j$-th neighbor of $i$. %As the state trajectories of all the metric systems are piecewise constant over $[0,\tau]$, hence $\nu_i$ and $\nu_i'$ are also piecewise constant over $[0,\tau]$, which allows us to interpret $\vmetric$ as the distance between $\nu_i$ and $\nu_i'$ in any given time interval.

If  $\set{S_i}_{i\in I}$ are compatible for composition, we define the \emph{composition} of any subset $I'\subseteq I$ of systems as the metric system $\llbracket S_i \rrbracket_{i\in I'}=(X, U, \mathcal{U}_{\tau}, W, \mathcal{W}_{\tau}, \delta)$ 
s.t.\
$X = \prod_{i\in I'} {X_i}$, 
$U = \prod_{i\in I'} {U_i}$, and 
$W = \prod_{j\in \nbr_{\mathcal{I}}(I')} {X_j}$, 
where $\mathcal{U}_{\tau}$ and $\mathcal{W}_{\tau}$ are defined over $U$ and $W$, respectively, as in \eqref{equ:def:mathcalUW}.
In analogy to \eqref{equ:defvmetric_composed} we equip the composed state space $X$ with the metric $d(x,x')=\norm{\prod_{j\in I'} \set{d_j(x_j,x_j')} }$.
The composed transition function is defined as 
$\delta_{\tau}(x,\mu,\nu) = \prod_{i\in I'} \{\delta_{\tau,i}(x_i, \mu_i, \nu_i^c\times\nu_i^e)\}$
where $x=\prod_{i\in I'} \set{x_i}$, $\mu=\prod_{i\in I'} \set{\mu_i}$, 
$\nu=\prod_{i\in I'} \set{\nu_i^e}$, and
$\nu_i^c=\prod_{j\in\nbr_{\mathcal{I}'}(i)} \set{\xi_j}$
% $\nu_i^e=\nu|_{W_i^e}$.
with $\xi_j$ being the continuous time evolution of $x_j$.
It follows immediately from this construction that the composed system $\llbracket S_i \rrbracket_{i\in I'}$ is again a metric system, with metric $d$. 
We extend the metric $\vmetric$ to the set $W$ by substituting $\nbr_{\mathcal{I}}(i)$ by $\nbr_{\mathcal{I}}(I')$ in \eqref{equ:defvmetric_composed}.

Intuitively, 
the composition of a set of compatible metric systems gives the joint dynamics of the network.
When we pick a subset of systems $I'\subseteq I$, the incoming edges from $\nbr_{\mathcal{I}}(I')$
become external disturbances for the composed subsystem. 
Observe that our approach is modular. We can first compose different disjoint sets of subsystems before composing the resulting systems together. 
Our definition of system composition is illustrated by the following example.

\begin{example}\label{ex}
\begin{figure}
\centering
 \begin{tikzpicture}
	\newcommand{\LBox}{0.9};	% Length of each box
	\newcommand{\HBox}{0.9};	% Height of each box
	
	\newcommand{\LArr}{0.7};		% Length of free-ended arrows
%	\newcommand{\x}{0.1};		% Offset between box and two-sided dashed arrows
%	\newcommand{\y}{0.5};		% Distance between R_{-,-} and \norm{}..\leq \varepsilon_{-} arrows

	% Coordinates of SW corners of the boxes
	\coordinate	(1SW)	at	(0,0);
	\coordinate	(2SW)	at	(3,0);
	\coordinate	(3SW)	at	(6,0);
	
	\draw	(1SW)	rectangle	($ (1SW) + (\LBox,\HBox) $)	node[pos=0.5]	{$S_1$};
	\draw	(2SW)	rectangle	($ (2SW) + (\LBox,\HBox) $)	node[pos=0.5]	{$S_2$};
	\draw	(3SW)	rectangle	($ (3SW) + (\LBox,\HBox) $)	node[pos=0.5]	{$S_3$};
	
	\draw[->]	($ (1SW) + (-\LArr,0.5*\HBox) $)	--	($ (1SW) + (0,0.5*\HBox) $)	node[pos=0.1,above]	{$\mu_1$};
	\draw[->]	($ (2SW) + (-\LArr,0.75*\HBox) $)	--	($ (2SW) + (0,0.75*\HBox) $)	node[pos=0.1,above]	{$\mu_2$};
	\draw[->]	($ (1SW) + (\LBox,0.5*\HBox) $)	--	($ (2SW) + (0,0.5*\HBox) $)	node[pos=0.2,above]	{$\xi_1$};
	\draw[->]	($ (3SW) + (-\LArr,0.75*\HBox) $)	--	($ (3SW) + (0,0.75*\HBox) $)	node[pos=0.1,above]	{$\mu_3$};
	\draw[->]	($ (3SW) + (\LBox,0.5*\HBox) $)	--	($ (3SW) + (\LBox+\LArr,0.5*\HBox) $)	node[pos=0.4,above]	{$\xi_3$}	--	($ (3SW) + (\LBox+\LArr,-0.5*\HBox) $)	--	($ (2SW) + (-\LArr,-0.5*\HBox) $)	--	($ (2SW) + (-\LArr,0.25*\HBox) $)	--	($ (2SW) + (0,0.25*\HBox) $);
	\draw[->]	($ (2SW) + (\LBox,0.25*\HBox) $)	--	($ (3SW) + (0,0.25*\HBox) $)	node[pos=0.4,above]	{$\xi_2$};
\end{tikzpicture}
 \caption{Network of metric systems containing cycles as discussed in Example~\ref{ex}.}\label{fig a}
 \end{figure}
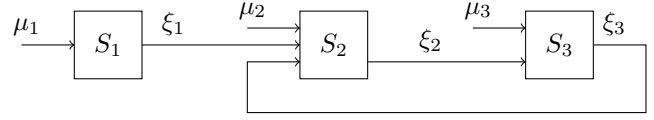
 \begin{figure}
\centering
 \begin{tikzpicture}
	\newcommand{\LBox}{1};	% Length of each box
	\newcommand{\HBox}{0.9};	% Height of each box
	
	\newcommand{\LArr}{1};		% Length of free-ended arrows
%	\newcommand{\x}{0.1};		% Offset between box and two-sided dashed arrows
%	\newcommand{\y}{0.5};		% Distance between R_{-,-} and \norm{}..\leq \varepsilon_{-} arrows

	% Coordinates of SW corners of the boxes
	%\coordinate	(1SW)	at	(0,0);
	\coordinate	(12SW)	at	(3,0);
	\coordinate	(3SW)	at	(7,0);
	
	%\draw	(1SW)	rectangle	($ (1SW) + (\LBox,\HBox) $)	node[pos=0.5]	{$S_1$};
	\draw	(12SW)	rectangle	($ (12SW) + (2*\LBox,2*\HBox) $)	node[pos=0.5]	{$\llbracket S_i \rrbracket_{i\in \set{1,2}}$};
	\draw	(3SW)	rectangle	($ (3SW) + (\LBox,\HBox) $)	node[pos=0.5]	{$S_3$};
	
	\draw[->]	($ (12SW) + (-\LArr,0.5*2*\HBox) $)	--	($ (12SW) + (0,0.5*2*\HBox) $)	node[pos=0.1,above]	{$\mu_2$};
	\draw[->]	($ (12SW) + (-\LArr,0.9*2*\HBox) $)	--	($ (12SW) + (0,0.9*2*\HBox) $)	node[pos=0.1,above]	{$\mu_1$};
	%\draw[->]	($ (1SW) + (\LArr,0.5*\HBox) $)	--	($ (2SW) + (0,0.5*\HBox) $)	node[pos=0.2,above]	{$\xi_1$};
	\draw[->]	($ (3SW) + (-\LArr,0.75*\HBox) $)	--	($ (3SW) + (0,0.75*\HBox) $)	node[pos=0.1,above]	{$\mu_3$};
	\draw[->]	($ (3SW) + (\LArr,0.5*\HBox) $)	--	($ (3SW) + (\LBox+\LArr,0.5*\HBox) $)	node[pos=0.4,above]	{$\xi_3$}	--	($ (3SW) + (\LBox+\LArr,-0.5*\HBox) $)	--	($ (12SW) + (-\LArr,-0.5*\HBox) $)	--	($ (12SW) + (-\LArr,0.1*2*\HBox) $)	--	($ (12SW) + (0,0.1*2*\HBox) $);
	\draw[->]	($ (12SW) + (2*\LBox,0.25*\HBox) $)	--	($ (3SW) + (0,0.25*\HBox) $)	node[pos=0.4,above]	{$\xi_2$};
\end{tikzpicture} 
 \caption{Composition of the subsystem $\set{S_i}_{i\in \set{1,2}}$ within the network depiced in Fig.~\ref{fig a}.}\label{fig b}
\end{figure}
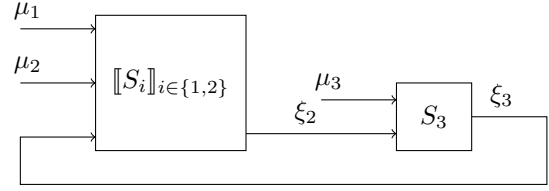
Consider the following three systems
\begin{align}
% 	\Sigma_1: \ \dot{x}_1 &= f_1(x_1,u_1)\\
% 	\Sigma_2: \ \dot{x}_2 &= f_2(x_2,u_2,(x_1,x_3))\\
% 	\Sigma_3: \	\dot{x}_3 &= f_3(x_3,u_3,x_1)
%	\Sigma_1: \ \dot{x}_1 &= f_1(x_1,u_1)\\
%	\Sigma_2: \ \dot{x}_2 &= f_2(x_2,u_2,w_2)\\
%	\Sigma_3: \	\dot{x}_3 &= f_3(x_3,u_3,w_3)
	S_i = (X_i,U_i,\mathcal{U}_{\tau,i},W_i,\mathcal{W}_{\tau,i},\delta_{\tau,i}) \qquad \text{for } i\in\set{1,2,3}.
\end{align}
%with states $x_1 \in X_1$, $x_2\in X_2$, $x_3\in X_3$, control inputs $u_1\in U_1$, $u_2 \in U_2$, $u_3\in U_3$ and disturbances $w_2\in W_2$ and $w_3\in W_3$. 
The index set and the interconnection relation are given by $I=\set{1,2,3}$ and $\mathcal{I} = \set{(1,2),(2,3),(3,2)}$, respectively, and the sets of neighbors are defined by $\nbr_{\mathcal{I}}(1) = \emptyset$, $\nbr_{\mathcal{I}}(2) = \set{1,3}$ and $\nbr_{\mathcal{I}}(3) = \set{2}$. The systems $\set{\Sigma_i}_{i\in I}$ are compatible for composition w.r.t.\ $I$ if $W_1 = \{0\}$, $W_2 = X_1\times X_3$ and $W_3 = X_2$. In this case the schematic representation of this network of systems is given in Fig.~\ref{fig a}.

Now assume that $\set{S_i}_{i\in I}$ are compatible and consider the composition of system $S_1$ and $S_2$, i.e. $\llbracket S_i \rrbracket_{\set{1,2}}=(X,U,\mathcal{U}_\tau,W,\mathcal{W}_\tau,\delta_\tau)$. This composition has the interconnection relation $\mathcal{I}' = \set{(1,2)}$ and the global set of neighbors $\nbr_{\mathcal{I}}(I') = \set{3}$. The coupling and external disturbance spaces are given by
$W_1^c = \{0\}$, $W_1^e = \{0\}$,
$W_2^c = X_1$ and $W_2^e = X_3$. The remaining sets are given by $X = X_1\times X_2$, $U = U_1\times U_2$, and $W = X_3$. Given some $x = (x_1,x_2)\in X$, $\mu = (\mu_1,\mu_2)\in \mathcal{U}_\tau$ and $\nu = \xi_3 \in \mathcal{W}_\tau$ (the continuous time version of $x_3$), the transition relation is given by $\delta_\tau(x,\mu,\nu) = (\delta_{\tau,1}(x_1,\mu_1,\SEZS{0}),\delta_{\tau,2}(x_2,\mu_2,(\xi_1,\xi_3)))$. By substituting system $S_1$ and $S_2$ by its composition $\llbracket S_i \rrbracket_{i\in\set{1,2}}$ we obtain the network shown in Fig.~\ref{fig b}.

\end{example}

% \begin{figure}
% 	\centering
% 	\subfigure[The full network consisting of the compatible control systems $\set{\Sigma_i}_{i\in \set{1,2,3}}$ discussed in Example~\ref{ex}.]{
% 		\input{fig_network_cs}
% 		\label{subfig:I}}\\
% 	\subfigure[Resulting network when replacing control systems $\Sigma_1$ and $\Sigma_2$ in Figure~\ref{subfig:I} by their composition $\ContSubSys{\set{1,2}}$.]{
% 		\input{fig_network_cs_composed}
% 		\label{subfig:I'}}
% 	\caption{Network of systems discussed in Example~\ref{ex}; in general the network could have cycles.}
% \label{figure:network}
% \end{figure}

% \subsection{Problem Statement}

% \todo

\section{Disturbance Bisimulation}\label{sec:DisturbanceBisimulation}
Before formally defining disturbance bisimulation, we want to motivate its need for \emph{compositional} abstraction-based controller synthesis.

\begin{figure*}
\centering
 \begin{tikzpicture}
	\newcommand{\LBox}{1};	% Length of each box
	\newcommand{\HBox}{1};	% Height of each box
	
	\newcommand{\LArr}{1};		% Length of free-ended arrows
	\newcommand{\x}{0.1};		% Offset between box and two-sided dashed arrows
	\newcommand{\y}{0.7};		% Distance between R_{-,-} and \norm{}..\leq \varepsilon_{-} arrows

	% Coordinates of SW corners of the boxes
	\coordinate	(1SW)	at	(-1,0);
	\coordinate	(h1SW)	at	(-1,-2);
	\coordinate	(2SW)	at	(4.5,0);
	\coordinate	(h2SW)	at	(4.5,-2);
	\coordinate	(12SW)	at	(11,0);
	\coordinate	(h12SW)	at	(11,-2);
	
	\draw	(1SW)	rectangle	($ (1SW) + (\LBox,\HBox) $)	node[pos=0.5]	{$S_1$};
	\draw	(h1SW)	rectangle	($ (h1SW) + (\LBox,\HBox) $)	node[pos=0.5]	{$\hat{S}_1$};
	\draw	(2SW)	rectangle	($ (2SW) + (\LBox,\HBox) $)	node[pos=0.5]	{$S_2$};
	\draw	(h2SW)	rectangle	($ (h2SW) + (\LBox,\HBox) $)	node[pos=0.5]	{$\hat{S}_2$};
	\draw	(12SW)	rectangle	($ (12SW) + (2*\LBox,\HBox) $)	node[pos=0.5]	{$\llbracket S_i \rrbracket_{i\in\set{1,2}}$};
	\draw	(h12SW)	rectangle	($ (h12SW) + (2*\LBox,\HBox) $)	node[pos=0.5]	{$\llbracket \hat{S}_i \rrbracket_{i\in\set{1,2}}$};
	
	\draw[->]	($ (1SW) + (-\LArr,0.5*\HBox) $)	--	($ (1SW) + (0,0.5*\HBox) $)	node[pos=0.1,above]	{$u_1$};
	\draw[->]	($ (h1SW) + (-\LArr,0.5*\HBox) $)	--	($ (h1SW) + (0,0.5*\HBox) $)	node[pos=0.1,above]	{$\hat{u}_1$};
	\draw[->]	($ (2SW) + (-\LArr,0.8*\HBox) $)	--	($ (2SW) + (0,0.8*\HBox) $)	node[pos=0.1,above]	{$u_2$};
	\draw[->]	($ (h2SW) + (-\LArr,0.8*\HBox) $)	--	($ (h2SW) + (0,0.8*\HBox) $)	node[pos=0.1,above]	{$\hat{u}_2$};
	\draw[->]	($ (12SW) + (-\LArr,0.5*\HBox) $)	--	($ (12SW) + (0,0.5*\HBox) $)	node[pos=0.1,below]	{$u_2$};
	\draw[->]	($ (12SW) + (-\LArr,0.75*\HBox) $)	--	($ (12SW) + (0,0.75*\HBox) $)	node[pos=0.1,above]	{$u_1$};
	\draw[->]	($ (h12SW) + (-\LArr,0.25*\HBox) $)	--	($ (h12SW) + (0,0.25*\HBox) $)	node[pos=0.1,below]	{$\hat{u}_2$};
	\draw[->]	($ (h12SW) + (-\LArr,0.75*\HBox) $)	--	($ (h12SW) + (0,0.75*\HBox) $)	node[pos=0.1,above]	{$\hat{u}_1$};
	
	\draw[->]	($ (1SW) + (\LBox,0.23*\HBox) $)	--	($ (2SW) + (0,0.23*\HBox) $)	node[pos=0.2,above]	{$x_1$} node[pos=0.8,above] {$w_2$};
	\draw[->]	($ (h1SW) + (\LBox,0.23*\HBox) $)	--	($ (h2SW) + (0,0.23*\HBox) $)	node[pos=0.2,below]	{$\hat{x}_1$} node[pos=0.8,above] {$\hat{w}_2$};
	\draw[->]	($ (2SW) + (\LBox,0.5*\HBox) $)	--	($ (2SW) + (\LBox+\LArr,0.5*\HBox) $)	node[pos=0.9,above]	{$x_2$};
	\draw[->]	($ (h2SW) + (\LBox,0.5*\HBox) $)	--	($ (h2SW) + (\LBox+\LArr,0.5*\HBox) $)	node[pos=0.9,below]	{$\hat{x}_2$};
	\draw[->]	($ (12SW) + (2*\LBox,0.5*\HBox) $)	--	($ (12SW) + (2*\LBox+\LArr,0.5*\HBox) $)	node[pos=0.9,above]	{$x_2$};
	\draw[->]	($ (h12SW) + (2*\LBox,0.5*\HBox) $)	--	($ (h12SW) + (2*\LBox+\LArr,0.5*\HBox) $)	node[pos=0.9,below]	{$\hat{x}_2$};
	
	\draw[<->,dashed]	($ (1SW) + (0.5*\LBox,-\x) $)	--	($ (h1SW) + (0.5*\LBox,\HBox+\x) $)	node[pos=0.5,left]	{$R_{\varepsilon_10}$};
	\draw[<->,dashed]	($ (2SW) + (0.5*\LBox,-\x) $)	--	($ (h2SW) + (0.5*\LBox,\HBox+\x) $)	node[pos=0.5,left]	{$R_{\varepsilon_2\tilde{\varepsilon}_2}$};
	\draw[<->,dashed]	($ (12SW) + (0.5*\LBox,-\x) $)	--	($ (h12SW) + (0.5*\LBox,\HBox+\x) $)	node[pos=0.5,left]	{$R_{\varepsilon\tilde{\varepsilon}}$};
	
	\draw[<->,dashed]	($ (1SW) + (\LBox+\y,0.23*\HBox-\x) $)	--	($ (h1SW) + (\LBox+\y,0.23*\HBox+\x) $)	node[right,pos=0.5,align=center]	{$\norm{x_1-\hat{x}_1}$\\$ \leq \varepsilon_1 =: \tilde{\varepsilon}_2$};
	\draw[<->,dashed]	($ (2SW) + (\LBox+\y,0.5*\HBox-\x) $)	--	($ (h2SW) + (\LBox+\y,0.5*\HBox+\x) $)	node[right,pos=0.5,align=center]	{$\norm{x_2-\hat{x}_2}$\\$ \leq \varepsilon_2$};
	
	\draw[->, line width=0.7mm]	($ (h2SW) + (\LBox+4.2*\y,1.5*\HBox) $)	--	($ (h2SW) + (\LBox+4.2*\y+0.6*\LArr,1.5*\HBox) $);
\end{tikzpicture}
 \caption{Illustration of the compositional abstraction of a simple network of metric systems using disturbance bisimulation, as formalized in \Thm~\ref{thm:composition of ADB MS}.}\label{fig_schematic}
\end{figure*}
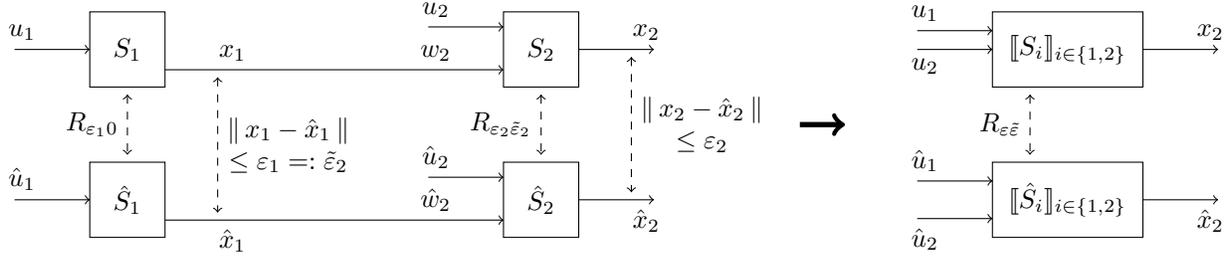

In the (monolithic) abstraction-based controller synthesis framework, a metric system $S$ is abstracted to a finite state metric system $\hat{S}$ s.t.\ a binary relation holds between the state space of the two which ensures that controllers synthesized for $\hat{S}$ can be refined to controllers for $S$. 
If disturbances are present in the system, these relations can be explained as follows. Consider two systems
$S$ and $\hat{S}$ that are \emph{approximately bisimilar} (as e.g.\ used in \cite{PolaPB14}). This relation requires that whatever input $\mu$ was chosen for $S$ (resp. $\hat{S}$) by its controller and whatever disturbance $\nu$ is currently present in $S$ (resp. $\hat{S}$), there exists a way to ensure that $\mu'$ and $\nu'$ can be chosen for $\hat{S}$ (resp. $S$), s.t.\ states which where initially $\varepsilon$-close are also $\varepsilon$-close at the next sampling instance (after applying these input and disturbance trajectories). While the assumption on choosing $\mu'$ appropriately can be justified by a careful controllers synthesis, it is unrealistic to assume that the choice of the disturbance signal for the second system is under the system designers control. 

Intuitively, controller synthesis in the presence of disturbances requires a relation where $\mu$ and $\mu'$ must be picked by both controllers s.t. that trajectories 
stay close for \emph{all possible} disturbances in both systems. %\KM{Not sure if ``all possible disturbance'' is the right phrase to use. Isn't it stronger than saying ``for all control input and disturbance pairs of one system, there exist one control and disturbance pair for the other s.t. ...''. Anne? }
While \emph{approximate alternating bisimulation} requires a different quantifier alternation, it also does not capture the above intuition as disturbances are still existentially quantified. In \cite{pola2009symbolic}, where approximately alternating bisimulations are used for controller synthesis, this problem is circumvented by assuming that only the system $S$ is subject to disturbances and the disturbance space of the abstraction $\hat{S}$ can be engineered in a way that the given relation is automatically fulfilled for all present disturbances. 
Unfortunately, this approach is not applicable to compositional abstraction as the disturbance signals of the abstractions are given by the abstract state trajectories of neighboring components and can therefore not be freely chosen.

Consider for example the network of metric systems $\set{S_{i}}_{i\in \set{1,2}}$ and their abstractions $\set{\hat{S}_{i}}_{i\in \set{1,2}}$ depicted in Fig.~\ref{fig_schematic}.
% . Furthermore, recall that we want to use disturbance bisimulations in an abstraction based controller synthesis framework where $\hat{S}$ needs to be finite state to enable automata-theoretic controller synthesis techniques. In particular, this implies that the abstract metric systems $\hat{S}_i$ evolve only via piecewise constant state trajectories.
% As in a network of metric systems (as defined in \Sec~\ref{sec:Comp:Metric}) disturbance trajectories of one system are actually the state trajectories of other components, having continuous (resp. piecewise constant) state trajectories implies that the disturbance trajectories in the network are also continuous (resp. piecewise constant). 
% This implies that 
The disturbance signal $\hat{\nu}_2$ (the continuous time version of $\hat{w}_2=\hat{x}_1$) applied to $\hat{S}_2$ is the piecewise constant state trajectory $\hat{\xi}_1$ of $\hat{S}_1$ and the disturbance signal $\nu_2$ (the continuous time version of $w_2$) applied to $S_2$ is the continuous state trajectory $\xi_1$  of $S_1$. Hence, both signals are provided by $S_1$ and $\hat{S}_1$ which are assumed to be controlled independently of $S_2$ and $\hat{S}_2$.  However, as both disturbance signals are the state trajectories of related systems we know that at sampling instances, $\nu_2$ and $\hat{\nu}_2$ are $\varepsilon_1$-close. Hence, we can use this knowledge about the mismatch of disturbance trajectories in the relation, as shown in the following formal definition.

% We therefore propose disturbance bisimulation as a binary relation where the systems only take turns in picking inputs and those inputs have to steer both systems in a way that for all disturbance trajectories which are \emph{close} in the described sense, state trajectories stay close as well. With this we avoid the implicit assumption that disturbances can be chosen by a system. While it would make sense in our setting to measure the closedness of disturbance trajectories over the intersampling period, this would require future knowledge over the behavior of other components trajectory when picking a (piecewise constant) control input and hence we refrain from doing this. However, we will see that when constructing abstract metric systems from control systems in \REFsec{} we need to make sure that the distance between trajectories in intersampling periods stays bounded.

% Using the above intuition, we define disturbance bisimulation as follows.

\begin{definition} \label{def:DisturbanceBisimulation}
Let $\Sys{1}$ and $\Sys{2}$ be two metric systems, with state-spaces $\State{1},\State{2}\subseteq\State{}$ %\subseteq\mathbb{R}^n$ 
and disturbance sets
$W_{1},W_{2}\subseteq W\subseteq \mathbb{R}^p$.
Furthermore, let $\State{}$ admit the metric $d:\State{}\times\State{}\rightarrow\mathbb{R}_{\geq 0}$ and $W$ admit the vector-valued metric
$\vmetric:W\times W\rightarrow\mathbb{R}_{\geq 0}^r$, $1\leq r\leq p$.
% Given the preliminaries of \Def~\ref{def:approx. bisimulation} and a sampling time $\tau\in \mathbb{R}_{>0}$, a
A binary relation $R \subseteq X_{1}\times X_{2}$ is a \emph{disturbance bisimulation with parameters $(\varepsilon,\tilde{\varepsilon})$} where $\varepsilon\in \poszreal{}$ and $\tilde{\varepsilon}\in\mathbb{R}_{\geq 0}^r$,
iff for each $(x_{1},x_{2})\in R$:
	\begin{enumerate}[(a)]
\item $d(x_{1},x_{2})\leq\varepsilon$; 
\item for every $\IPETraj{1}\in \mathcal{U}_1$ there exists a $\IPETraj{2} \in \mathcal{U}_2$ such that for all $\IPITraj{2}\in \mathcal{W}_{\tau,2}$ and $\IPITraj{1} \in \mathcal{W}_{\tau,1}$ with
$\vmetric(\nu_1(0),\nu_2(0))\leq \tilde{\varepsilon}$, we have that
$(\delta_{\tau,1}(x_{1},\IPETraj{1},\IPITraj{1}), \delta_{\tau,2}(x_{2}, \IPETraj{2}, \IPITraj{2}))\in R$; and
\item for every $\IPETraj{2}\in \mathcal{U}_2$ there exists a $\IPETraj{1} \in \mathcal{U}_1$ such that for all $\IPITraj{1}\in \mathcal{W}_{\tau,1}$ and $\IPITraj{2} \in \mathcal{W}_{\tau,2}$ with
$\vmetric(\nu_1(0),\nu_2(0))\leq \tilde{\varepsilon}$, we have that
$(\delta_{\tau,1}(x_{1},\IPETraj{1},\IPITraj{1}), \delta_{\tau,2}(x_{2}, \IPETraj{2}, \IPITraj{2}))\in R$.
	\end{enumerate}
Two systems $\Sys{1}$ and $\Sys{2}$ are said to be \emph{disturbance bisimilar} with parameters $(\varepsilon,\tilde{\varepsilon})$ if
there is a disturbance bisimulation relation 
$R$ between $\Sys{1}$ and $\Sys{2}$ with parameters $(\varepsilon,\tilde{\varepsilon})$.
\end{definition}
% \SEZS{Notation $\delta_\tau$ or $\delta(\tau,...)$? $\mathcal{W}_{\tau,1}$?}

% Note that disturbance bisimulation is stronger than approximate disturbance bisimulation in the sense that all disturbance bisimulations are approximate disturbance bisimulations, but not the other way around.
% \AKS{Todo: give more details}

% Disturbance bisimulation can be intuitively understood as a two-player game between the systems $S_1$ and $S_2$, where the players take turns in picking and matching inputs as follows: in the first round, $S_1$ (conversely, $S_2$) picks a control input independently, and $S_2$ ($S_1$) tries to match by picking another control input. In the second round, it is $S_1$'s ($S_2$'s) turn to make an independent choice for picking disturbance inputs \emph{both for herself and $S_2$ ($S_1$), such that the disturbances are close to each other}. Thus, in contrast to alternating bisimulation relation introduced in \cite{pola2009symbolic}, the control never goes back to the starting player at the end of the second round. Having this intuition in mind, it is easy to see that disturbance bisimulation is a stronger relation than alternating bisimulation. Hence, if two systems are disturbance bisimilar with parameters $\paraeps{}$ they are also $\varepsilon$-approximately alternatingly bisimilar to each other.

% \subsection{Disturbance Bisimilar Compositions of Metric Systems}
As our first main result we show in the following theorem that disturbance bisimulation naturally extends from related components in a network to subsystems composed from them, which is also illustrated for a simple network in Fig.~\ref{fig_schematic}. 
\begin{theorem} \label{thm:composition of ADB MS}
Let $\set{S_{i}}_{i\in I}$ and $\set{\hat{S}_{i}}_{i\in I}$ be sets of compatible metric systems, s.t. for all $i\in I$, $S_{i}$ and $\hat{S}_{i}$ are disturbance bisimilar w.r.t. parameters $(\varepsilon_i,\tilde{\varepsilon}_i)$.
If
\begin{equation}\label{equ:tildevarepsilon 1}
\textstyle \tilde{\varepsilon}_i:=\prod_{j\in\nbr_{\mathcal{I}}(i)} \set{\varepsilon_j}
\end{equation}
%  Finally, for a subset $I'\subseteqI$ we define $\llbracket S_{i}\rrbracket_{i\in I'}=(X_{I'}, \mathcal{U}_{I'}, \mathcal{W}_{I'}, \delta_{I'})$ and $\llbracket \hat{S}_{i}\rrbracket_{i\in I'}=(\hat{X}_{I'}, \hat{\mathcal{U}}_{I'}, \hat{\mathcal{W}}_{I'}, \hat{\delta}_{I'})$ to be composed subsystem of the respective system set. 
 then for any given $I'\subseteq I$, the relation
 			\begin{align}\label{equ:eqclass_subsystems 1}
				\eqclass{} = &\lbrace (\transpose{ x^T_{1}\delimit\ldots\delimit x^T_{|I'|}} , \transpose{\hat{x}^T_{1}\delimit\ldots\delimit\hat{x}^T_{|I'|}})\in X_{I'} \times \hat{X}_{I'} \ |\notag\\
				&\quad(x_{i},\hat{x}_{i})\in\eqclass{i}, \forall i\in I'
				)\rbrace
			\end{align}
			% ******* Old stuff ******
%			\begin{align}\label{equ:eqclass_subsystems 1}
%				\eqclass{} = &\lbrace (\transpose{\hat{q}_{1}\delimit\ldots\delimit\hat{q}_{|I'|}},\transpose{ q_{1}\delimit\ldots\delimit q_{|I'|}})\in\hat{X}_{I'}\times X_{I'} \ |\notag\\
%				&\quad(\hat{q}_{i},q_{i})\in\eqclass{i}, \forall i\in I'
%				)\rbrace
%			\end{align} 
			% ***********************
		 is an approximate disturbance bisimulation between $\llbracket S_{i}\rrbracket_{i\in I'}$ and $\llbracket \hat{S}_{i}\rrbracket_{i\in I'}$ with parameters
		 	\begin{align*}%\label{equ:pickprecision}
% 		 		\tau,~ \ 
				\textstyle\varepsilon{} \textstyle=\norm{\prod_{i\in I'}\set{\varepsilon_{i}}} \text{ and }~ %\inv{\alphalow{I'}}\left(\sum_{i\in I'} c_i\alphalow{i}(\varepsilon_{i})\right),~
				\tilde{\varepsilon}\textstyle=\prod_{j \in \nbr_{\mathcal{I}}(I')} \set{\varepsilon_{j}}.
			\end{align*}
\end{theorem}

\begin{proof}
	%\KM{to finish}
	We prove all three parts of \Def~\ref{def:DisturbanceBisimulation} separately.\\
	\begin{inparaenum}[(a)]
		\item We pick a related tuple of trajectories $(x,\hat{x})\in\eqclass{}$ with $x = \transpose{ x^T_{1}\delimit\ldots\delimit x^T_{|I'|}}$ and $\hat{x} = \transpose{\hat{x}^T_{1}\delimit\ldots\delimit\hat{x}^T_{|I'|}}$. Then \eqref{equ:eqclass_subsystems 1} implies for all $i$, $(x_i,\hat{x}_i)\in \eqclass{i}$, which in turn gives $d_i(x_i,\hat{x_i})\leq \varepsilon_i$. %, which in turn gives $\Lypv{i}(\hat{q}_i,q_i)\leq\alphalow{i}(\varepsilon_{i})$. By \eqref{eqn:lyapunov condition 1}, we have $\norm{\hat{q}_i-q_i}\leq \varepsilon_i$ for all $i$.
		 This immediately gives $d(x,\hat{x}) = \norm{\prod_{j\in I'} \set{d_j(x_j,\hat{x}_j)} }
		 \leq \norm{\prod_{i\in I'}\set{\varepsilon_i}} = \varepsilon$.\\ %Now consider the global Lyapunov function $\Lypv{I'}$ constructed in 
		\item We pick the same related state tuple $(x,\hat{x})\in\eqclass{}$.
		%\KM{Why did you use $\eqclassC{}$ instead of $\eqclass{}$? (both here and in the beginning of (a))} 
		Note that the choice of $(x,\hat{x})$ automatically fixes the initial point of the coupling disturbances for the individual subsystems $\nu^c_i(0)$ and $\hat{\nu}^c_i(0)$ for $i \in I'$ s.t.
% 		${\nu^c_i(0)}= \prod_{j\in\nbr_{\mathcal{I}'}(i)}\lbrace \xi_j(0) \rbrace = \prod_{j\in\nbr_{\mathcal{I}'}(i)}\lbrace x_j \rbrace$ and ${\hat{\nu}^c_i(0)}= \prod_{j\in\nbr_{\mathcal{I}'}(i)}\lbrace \hat{\xi}_j(0) \rbrace =\prod_{j\in\nbr_{\mathcal{I}'}(i)}\lbrace \hat{x}_j \rbrace$.
		${\nu^c_i(0)}=  \prod_{j\in\nbr_{\mathcal{I}'}(i)}\lbrace x_j \rbrace$ and ${\hat{\nu}^c_i(0)}=\prod_{j\in\nbr_{\mathcal{I}'}(i)}\lbrace \hat{x}_j \rbrace$.
		As $(x_j,\hat{x}_j)\in \eqclass{j}$ we have $d(x_j,\hat{x}_j)\leq\varepsilon_j$. Using the definition of $\vmetric$ in \eqref{equ:defvmetric_composed} we therefore have $\vmetric({\nu^c_i(0)},{\hat{\nu}^c_i(0)})\leq\prod_{j \in \nbr_{\IntCon'}(i)} \set{\varepsilon_j}$.
		Now pick $\mu = \transpose{\mu_1\delimit\ldots\delimit \mu_{|I'|}}\in\mathcal{U}_\tau$, and $\nu\in\mathcal{W}_\tau,~\hat{\nu}\in\hat{\mathcal{W}}_{\tau}$ s.t. $\vmetric({\nu(0)},{\hat{\nu}}(0))\leq\tilde{\varepsilon}=\prod_{j \in \nbr_{\mathcal{I}}(I')} \set{\varepsilon_{j}}$. Recall from the definition of the composed metric systems that $\nu=\prod_{i\in I'} \set{\nu_i^e}$. 
		With this, it follows that 
% 			
% 		$\llbracket S_i\rrbracket_{i\in I'}$ and $\llbracket\hat{S}_i\rrbracket_{i\in I'}$ that we have $\nu=\prod_{j\in \nbr_{\mathcal{I}}(I')}\set{\xi_j}$ and ${\hat{\nu}}=\prod_{j\in \nbr_{\mathcal{I}}(I')}\set{\hat{\xi}_j}$ with $\xi_j\in \mathcal{X}_{\tau,j}$ and $\hat{\xi}_j\in\hat{\mathcal{X}}_{\tau,j}$ for $j\in I\setminus I'$. Using \eqref{equ:defvmetric_composed} we therefore have
% 		\begin{align*}
% 		 \vmetric(\nu(0),{\hat{\nu}}(0))
% 		 &=\prod_{j\in \nbr_{\mathcal{I}}(I')}\set{ d_j(\xi_j(0)-\hat{\xi}_j(0))}\\
% 		 &=\prod_{j\in \nbr_{\mathcal{I}}(I')}\set{ d_j(x_j,\hat{x}_j) }
% 		 \leq\prod_{j \in \nbr_{\mathcal{I}}(I')} \set{\varepsilon_j}= \tilde{\varepsilon}.
% 		\end{align*}
% 		Moreover, using ${\nu_i^e}= \prod_{j\in\nbr_{\IntCon\setminus\IntCon'}(i)} \set{\xi_j}$ and ${\hat{\nu}_i^e}= \prod_{j\in\nbr_{\IntCon\setminus\IntCon'}(i)} \set{\hat{\xi}_j}$, we have 
% 		
		$\vmetric({\nu^e_i}(0),{\hat{\nu}^e_i}(0))\leq\prod_{j \in \nbr_{\IntCon\setminus\IntCon'}(i)} \set{\varepsilon_j}$. Hence
% 		
% 		Now recall from the proof of \Thm~\ref{thm:sim-appx} that we can always pick the external disturbances $\hat{\nu}^e_i$ and $\nu^e_i$ with $w_{\nu^e_i}=\prod_{j\in\nbr_{\IntCon\setminus\IntCon'}(i)}\lbrace x_j \rbrace$ and $w_{\hat{\nu}^e_i}=\prod_{j\in\nbr_{\IntCon\setminus\IntCon'}(i)}\lbrace \hat{x}_j \rbrace$ for some $x_j\in\State{j}$ and $\hat{x}_j\in\StateA{j}$ s.t. $\norm{x_j-\hat{x}_j}\leq\eta_j\leq\varepsilon_j$.
% Using the triangular inequality again, we therefore have $\norm{\nu^e_i-\hat{\nu}^e_i}\leq\sum_{j \in \nbr_{\IntCon\setminus\IntCon'}(i)}\varepsilon_j$ and hence
\begin{multline*}
			\vmetric( {\nu_i}(0), {\hat{\nu}_i}(0) ) = \vmetric\left( \begin{bmatrix} {\nu^c_i}(0)\\ {\nu^e_i}(0) \end{bmatrix},\begin{bmatrix} {\hat{\nu}^c_i}(0)\\ {\hat{\nu}^e_i}(0) \end{bmatrix} \right)\leq 
			 \prod_{j\in \nbr_{\IntCon}(i)} \set{\varepsilon_j}.% =: \tilde{\varepsilon}_i, %\label{equ:proof:wi}
		\end{multline*}
		Using \eqref{equ:tildevarepsilon 1} we therefore have $\vmetric( {\nu_i}(0), {\hat{\nu}_i}(0) )\leq \tilde{\varepsilon}_i$.
		 With these local disturbance vectors and the fact that $S_{i}$ and $\hat{S}_{i}$ are approximately disturbance bisimilar w.r.t. $(\varepsilon_i,\tilde{\varepsilon}_i)$ it follows immediately from \Def~\ref{def:DisturbanceBisimulation} (b) that for any local control input $\mu_i$ there exits $\hat{\mu}_i$ such that $(\delta_{\tau,i}(x_i,\mu_i,\nu_i),\hat{\delta}_{\tau,i}(\hat{x}_i,\hat{\mu}_i,\hat{\nu}_i))\in\eqclass{i}$ for $i\in I'$. Then by \eqref{equ:eqclass_subsystems 1}, it follows that $(\delta_{\tau}(x,\mu,\nu),\hat{\delta}_{\tau}(\hat{x},\hat{\mu},\hat{\nu}))\in\eqclass{}$.\\
% 		 
% 		 We now collect all local input and disturbance vectors in the composed vectors $\mu = \transpose{\mu_1\delimit\ldots\delimit \mu_{|I'|}}$, $\hat{\mu} = \transpose{\hat{\mu}_1\delimit\ldots\delimit\hat{\mu}_{|I'|}}$, $\hat{\nu} = \transpose{ \hat{\nu}^e_1\delimit\ldots\delimit\tilde{\nu}^e_{|I'|}}$ and $\nu = \transpose{\nu^e_1\delimit\ldots\delimit \nu^e_{|I'|}}$ and observe that 
% 		 $\norm{\nu^e-\hat{\nu}^e} \leq 
%  		\sum_{i\in I'} \norm{\nu^e_i-\hat{\nu}^e_i} \leq 
%  		\sum_{i\in I'} \sum_{j \in \nbr_{\IntCon \setminus \IntCon'}(i)}\varepsilon_j = 
%  		\tilde{\varepsilon}$.
 		%Using the definition of the composed transition functions from \Sec~\ref{sec:Comp:Metric} and \eqref{equ:eqclass} it follows immediately that 
 		%$(\deltaA{}(\hat{q},\hat{\mu},\hat{\nu}),\delta_{\tau}(q,\mu,\nu))\in\eqclass{}$, what proves the statement.\\
%  		 
		 \item The other direction can be shown based on the same reasoning as for part (b) and is therefore omitted.
		 \end{inparaenum}
	\end{proof}

\section{Control Systems}\label{sec:ControlSystems}

We start the second part of the paper by introducing some necessary preliminaries on control systems and their stability.

\subsection{Preliminaries}

A \emph{control system} $\Sigma = (X, U, W, \mathcal{U}, \mathcal{W}, f)$
consists of a state space $X$, an input space $U$, a disturbance space $W$, 
a set of input signals $\IPECurve{}$, a set of disturbance signals $\IPICurve{}$,
and 
a continuous state transition function
${\TransFunc{}}: {X\times U \times W}\rightarrow {X}$.
We assume
$X= \real{n}$, $U= \real{m}$ and $W=\real{p}$ to be normed Euclidean spaces.
% , and
%  to be a compact subset 
% of a normed Euclidean space
%(\SEZS{of dimensions $n$, $m$, and $p$, respectively})
% that contains the origin. 
Furthermore, we assume that the
sets 
$\IPECurve{}$ and $\IPICurve{}$ consist 
of measurable essentially bounded functions ${\mu}: {\poszreal} \rightarrow {U}$ and 
${\nu}: {\poszreal} \rightarrow {W}$, respectively,
and 
$f$ satisfies the following Lipschitz assumption:
there exists a constant $L > 0$ \st
$\norm{\TransFunc{}(x, u, w) - \TransFunc{}(y, u, w)} \leq L \norm{x - y}$
for all $x, y\in X$, $u\in U$, and $w\in W$, where $\norm{\cdot}$ is a norm.

%	A \emph{control system} $\Sigma = (\real{n}, U, W, \mathcal{U}, \mathcal{W}, f)$
%consists of a state space $\real{n}$, an input space $U$ and a disturbance space $W$ where
%$U\subseteq \real{p}$ and $W\subseteq \real{q}$ are compact subsets containing the origin,
%sets 
%$\IPECurve{}$ and $\IPICurve{}$, 
%of input and disturbance signals consisting
%of measurable essentially bounded functions ${\mu}: {\poszreal} \rightarrow {\IPESpace{}}$ and 
%${\nu}: {\poszreal} \rightarrow {\IPISpace{}}$, respectively,
%and 
%a continuous state transition function
%${\TransFunc{}}: {\StateSpace{}\times U \times W}\rightarrow {\StateSpace{}}$
%satisfying the following Lipschitz assumption:
%for every compact set $K\subset \real{n}$, there exists a constant $L > 0$ such that
%$\norm{\TransFunc{}(x, u, w) - \TransFunc{}(y, u, w)} \leq L \norm{x - y}$
%for all $x, y\in K$, $u\in U$, and $w\in W$.

A trajectory 
${\StateTraj{}{}}: {(a,b)}\rightarrow {\StateSpace{}}$ 
associated with the control system $\Sigma{}$ and 
signals $\IPETraj{}\in\IPECurve{}$ and $\IPITraj{}\in\IPICurve{}$ 
is an absolutely continuous curve  satisfying: 
\begin{equation}
\StateTrajDot{}(t) = \TransFunc{}(\StateTraj{}{}(t),\IPETraj{}(t),\IPITraj{}(t)) \label{eq:state trajectory}
% \OPTraj{}(t) &= \OPFunc{}(\StateTraj{}{}(t))
\end{equation}
for almost all $t\in (a,b)$.
Although we define trajectories over open intervals, we talk about trajectories $\StateTraj{}{}:[0,\tau] \rightarrow X$
for $\tau\in\posreal{}$, with the understanding that $\StateTraj{}{}$ is the restriction to $[0,\tau]$ of some trajectory
defined on an open interval containing $[0,\tau]$.
We denote by $\StateTraj{}{x\IPETraj{}\IPITraj{}}(\cdot)$ the solution of differential equation \eqref{eq:state trajectory} with initial condition $x$ and with input and disturbance signals $\IPETraj{}$ and $\IPITraj{}$, respectively. This solution is unique due to the Lipschitz continuity assumption on $f$.
	Thus $\StateTraj{}{x\IPETraj{}\IPITraj{}}(t)$ is the state reached by the 
	trajectory $\StateTraj{}{}$ starting from $x$ with input and disturbance signals $\IPETraj{}$ and $\IPITraj{}$.
A control system $\Sigma$ is \emph{forward complete}%\footnote{Sufficient and necessary conditions for a system
% 		to be forward complete can be found in \cite{sontag} for $\Sigma$ with $X=\real{n}$, $U=\real{m}$, and $W = \real{p}$. The conditions can be adapted to compact sets $U,W$ by freezing the state evolution \eqref{eq:state trajectory} at the boundaries of $X$.} 
		if every trajectory defined on an interval $(a,b)$ can be extended to an interval of the form $(a,\infty)$.

	If the control system $\Sigma$ is undisturbed, we define $W=\{0\}$. In this case we occasionally represent $\Sigma$ by the tuple $\Sigma = (X, U, \mathcal{U}, f)$ and use $f:X\times U\rightarrow X$ with the understanding that \eqref{eq:state trajectory} holds for $f$ for the zero trajectory $\IPITraj{}:{\poszreal}\rightarrow \set{0}$ whenever  $\dot{\xi}(t)=f(\xi(t),\mu(t))$ holds.

\subsection{Input-to-state Lyapunov functions}\label{sec:prelim:LF}
% Input-to-state Lyapunov functions will pave the way for the abstraction process. To define them we first introduce additional notation.

A continuous function \mbox{$\gamma:\mathbb{R}_{\geq 0}\rightarrow\mathbb{R}_{\geq 0}$} is said to belong to class 
$\mathcal{K}_{\infty}$ if it is strictly increasing, \mbox{$\gamma(0)=0$}, and
$\gamma(r)\rightarrow\infty$ as $r\rightarrow\infty$. 
A continuous function \mbox{$\beta:\mathbb{R}_{\geq 0}\times\mathbb{R}_{\geq 0}\rightarrow\mathbb{R}_{\geq 0}$} 
is said to belong to class $\mathcal{KL}$ if, for each fixed $s$, 
the map $\beta(r,s)$ belongs to class $\mathcal{K}_{\infty}$ with respect to $r$ and, 
for each fixed nonzero $r$, the map $\beta(r,s)$ is decreasing with respect to $s$ and $\beta(r,s)\rightarrow0$ as \mbox{$s\rightarrow\infty$}.
%Let $f:(a,b)\rightarrow\mathbb{R}^k$ be a piecewise continuous function which is also essentially 
%bounded to a region $B\subset\mathbb{R}^k$. 
%We define the supremum norm $\normi{f}$ of $f$ as
%$\normi{f} := \mathrm{max}\set{ |s| \mid \exists r\in (a,b).f(r)=s\wedge s\in B}$.

\begin{definition} \label{def:lyapunov function}
Given a control system $\Sigma$, a smooth function 
${V}:{X\times X}\rightarrow {\real{}}$ is said to be 
a \emph{$\dISS$ Lyapunov function} for $\Sigma$ if there exist 
$\lambdaL{}\in\posreal{}$ and $\kifunc$ functions $\alphalow{}$, $\alphahigh{}$, 
$\sigmau{}$, and $\sigmad{}$
\st for any $\state{},\state{}'\in X$, $\ipe{},\ipe{}'\in\IPE{}$, and
$\ipi{}, \ipi{}' \in \IPI{}$, the following holds:
	\begin{align}
		&\alphalow{}(\norm{\state{}-\state{}'})  \leq V(\state{},\state{}') \leq \alphahigh{}(\norm{\state{}-\state{}'})\quad\text{and} \label{eqn:lyapunov condition 1}\\
% \end{align}
% \begin{align}
&\partderiv{V}{\state{}}  \TransFunc{}(\state{},\ipe{},\ipi{}) + \partderiv{V}{\state{}'}\TransFunc{}(\state{}',\ipe{}', \ipi{}') \leq \nonumber\\
&\quad -\lambda V(\state{},\state{}') + \sigmau{}(\norm{\ipe{}-\ipe{}'}) + \sigmad{}(\norm{\ipi{} - \ipi{}'}). \label{eqn:lyapunov condition 2}
\end{align}
In this case we say that the control system $\Sigma$ \emph{admits} a Lyapunov function $V$, \emph{witnessed by} 
$\lambda$, $\alphalow{}$, $\alphahigh{}$, $\sigmau{}$, and $\sigmad{}$.
\end{definition}

%\AKS{w.r.t. the comments by the reviewers I find the term \enquote{signals} in the above definition missleading. I ussually associate a time function (from $\mathcal{W}$ and $\mathcal{U}$) with the term \enquote{signal}. I would simply remove the explaination and write:\\
%for any $\state{},\state{}'\in X$, $\ipe{},\ipe{}'\in\IPE{}$, and
% $\ipi{}, \ipi{}' \in \IPI{}$, the following holds:}

% While we will assume in this paper that we are given $\dISS$ Lyapunov functions,
% there is a characterization of the latter in terms of stability.
%
%\begin{definition} \label{def:dISS}

Existence of $\dISS$ Lyapunov functions is tightly connected with the control system $\Sigma$ being \emph{incrementally globally input-to-state stable} ($\dISS$) \cite{Angeli02,TabuadaBook}.
It is shown in \cite{Angeli02} that under mild assumptions on control systems %with $X = \real{n}$ 
the existence of a $\dISS$ Lyapunov function is equivalent to $\dISS$ stability.

%A control system $\Sigma{}$ is \emph{incrementally globally input-to-state stable} ($\dISS$) 
%if it is forward complete and there exist a 
%$\klfunc$ function $\beta$ and two $\kifunc$ functions $\rho_u$ and $\rho_d$ \st for any $t\in\poszreal{}$, 
%any $\state{},\state{}'\in\StateSpace{}$, and any $\IPETraj{},\IPETraj{}'\in\IPECurve{}$, the following inequality is satisfied:
%\begin{multline}
%|\!| ~\StateTraj{}{\state{}\IPETraj{}\IPITraj{}}(t) - \StateTraj{}{\state{}'\IPETraj{}'\IPITraj{}'}(t) |\!| 
% \leq \beta(\norm{\state{}-\state{}'},t)\\ + \rho_u(\normi{\IPETraj{}-\IPETraj{}'}) + \rho_d(\normi{\IPITraj{} - \IPITraj{}'}).
%\end{multline}
%%\end{definition}
%%
%Under mild assumptions, e.g., that $f(0,0,0) = 0$, and $U$ and $W$ are compact and \emph{convex}
%sets, the existence of a $\dISS$ Lyapunov function is equivalent to $\dISS$ stability
%\cite{Angeli02,TabuadaBook}. 

%\subsection{Well defined Control Systems}
%\AKS{Maybe choose a different name.}
%\begin{definition}\label{def:CS_welldefined}
% We call the control system $\Sigma$ \emph{well defined} if it admits a $\dISS$ Lyapunov function 
%	${\Lypv{}}$ witnessed by $\lambda$, $\alphalow{}$, $\alphahigh{}$, $\sigmau{}$, and $\sigmad{}$ s.t. 
%	there exists a $\KInf$ function $\gammafunc{}$ \st for any $\state{},\state{}',\state{}''\in\mathbb{R}^{n}$ it holds that
%		\begin{align}
%			\Lypv{}(\state{}',\state{}) - \Lypv{}(\state{}'',\state{}) \leq \gammafunc{}(\norm{\state{}'-\state{}''}).
%		\end{align} 
%\end{definition}

We further restrict the class of $\dISS$ Lyapunov functions by requiring the following property.
We assume there exists a $\KInf$ function $\gammafunc{}$ \st for any $\state{},\state{}',\state{}''\in\mathbb{R}^{n}$ it holds that
	\begin{equation}
	\label{eq:ISS_cond}
		\Lypv{}(\state{}',\state{}) - \Lypv{}(\state{}'',\state{}) \leq \gammafunc{}(\norm{\state{}'-\state{}''}).
	\end{equation} 
Note that this is a very mild assumption which is satisfied by most $\dISS$ Lyapunov functions in practice (e.g., quadratics, polynomials, and square roots).
%\AKS{Sadegh, can you add a comment why this is not restrictive (as in the stochatic paper)?}
%\SEZS{Anne: the stochastic case is a bit different. There we have a concave function bounding a convex function.
%I can't think of a usual function not satisfying the assumption.}

\section{Disturbance Bisimilar Symbolic Models for Control Systems}\label{sec:MonolithicAbstraction}

% \subsection{Disturbance Bisimilar Metric Systems induced by $\Sigma$}
% We adapting the usual construction of time-sampled and abstract metric systems from control systems presented in \cite{PolaGT08} and \cite{GirardPolaTabuada_2010}, 
This section adapts the construction of time-sampled and abstract metric systems from \cite{PolaGT08} and \cite{GirardPolaTabuada_2010} to the notion of disturbance bisimulation.
We start by defining a metric system as a time-sampled version of a control system.
% \AKS{Add references here, required by reviewer.}

% \KM{Add text on piecewise constant inputs and disturbances.}
\begin{definition}
\label{def:control-sys-disctime}
Given a control system $\Sigma = (X, U, \mathcal{U}, W, \mathcal{W}, f)$,
and a \emph{time-sampling parameter} $\tau \in \posreal{}$, the
\emph{discrete-time metric system} induced by $\Sigma$ is defined by
\begin{equation}
 \ContSysDT{\Sigma} = (X, U, \mathcal{U}_\tau, W, \mathcal{W}_\tau, \delta_\tau)
\end{equation}
s.t.\
$\mathcal{U}_{\tau}$ and $\mathcal{W}_{\tau}$ are defined over $U$ and $W$, respectively, as in \eqref{equ:def:mathcalUW} and
$\delta_\tau(x,\mu,\nu) = \StateTraj{}{x\mu\nu}(\tau)$. % f $\mu\in\IPECurveDT{}$ and $\nu\in\IPICurveDT{}$. 
We equip $X$ with the metric $\dist{x}{x'} \defeq \norm{x -x'}$.
\end{definition}

%%%%
To define an abstract metric system $\ContSysA{}{\Sigma}$ induced by $\Sigma$ which is disturbance bisimilar to $\ContSysDT{\Sigma}$ we need some notation to discretize the state, input, and disturbance spaces of $\Sigma$.

%For any $A\subseteq \real{n}$ and $\StateQnt{} > 0$, we define 
%$\DiscSpace{A}{\StateQnt{}}:=
%\set{ (a_1,\ldots, a_n) \in A \mid a_i = k\frac{2}{\sqrt{n}}\StateQnt{}, k\in\mathbb{Z}, i=1,\ldots, n}$.
%For $x\in\real{n}$ and $\lambda > 0$, let
%$\mathbb{B}_\lambda(x)$ denote the closed ball centered at $x$ and of radius $\lambda$.
%Note that for any $\lambda \geq \StateQnt{}$, the collection of sets
%$\set{\mathbb{B}_\lambda(q)}_{q\in \DiscSpace{\real{n}}{\StateQnt{}} }$ is a covering of $\real{n}$,
%that is, $\real{n} \subseteq \cup \set{\mathbb{B}_\lambda(q)\mid q\in \DiscSpace{\real{n}}{\StateQnt{}} }$.

For any $A\subseteq \real{n}$ and $\StateQnt{}$ with elements $\StateQnt{i}>0$, we define 
$\DiscSpace{A}{\StateQnt{}}:=
\set{ (a_1,\ldots, a_n) \in A \mid a_i = 2k\StateQnt{i}, k\in\mathbb{Z}, i=1,\ldots, n}.$
For $x\in\real{n}$ and vector $\lambda$ with elements $\lambda_i>0$, let
$\mathbb{B}_\lambda(x) = \set{x'\in\real{n} \mid \norm{x_i-x_i'}\leq \lambda_i}$ denote the closed rectangle centered at $x$.
Note that for any $\lambda \geq \StateQnt{}$ (element-wise), the collection of sets
$\mathbb{B}_\lambda(q)$ with $q\in \DiscSpace{\real{n}}{\StateQnt{}}$ is a covering of $\real{n}$,
that is, $\real{n} \subseteq \cup \set{\mathbb{B}_\lambda(q)\mid q\in \DiscSpace{\real{n}}{\StateQnt{}} }$.

We will use this insight to discretize the state and the input space of $\Sigma$ using discretization parameters $\eta$ and $\omega$, respectively. For the disturbance space $W$ we allow the discretization of $W$ to be predefined. Intuitively, this models the fact that  discretizing all state spaces $X_i$ directly discretizes the disturbance spaces $W_i$ in a network of control systems, formally defined in \Sec~\ref{sec:NetworkControlSystems}. We therefore make the following general assumptions on the discretizaion of $W$.

\begin{assumption}\label{ass:Wt}
 Let $\Sigma = (X, U, \mathcal{U}, W,\allowbreak \mathcal{W}, f)$ be a control system.
We assume there exists a countable set $\Wt\subseteq W$, a vector $\tilde{\varepsilon} \in \mathbb{R}_{\geq 0}^r$, and
$\Wt$ is equipped with a (possibly vector-valued) metric $\vmetric:W\times W\rightarrow\mathbb{R}_{\geq 0}^r,~1\leq r\leq p$ 
s.t.\ for all $w\in W$ there exists a $\wt\in\Wt$ s.t.
       \begin{equation} \label{equ:Vtildeeps}
        \vmetric(w,\wt)\leq\tilde{\varepsilon} \quad
        \text{and}
        \quad \norm{w-\wt}\leq\norm{\vmetric(w,\wt)}.
       \end{equation}
%        for some vector $\tilde{\varepsilon} \in \mathbb{R}_{\geq 0}^r$.
\end{assumption}

Using this assumption we formally define the abstract metric system $\ContSysA{}{\Sigma}$ induced by $\Sigma$ as follows. 

\begin{definition} \label{def:abstract system}
       Let $\Sigma = (X, U, \mathcal{U}, W,\allowbreak \mathcal{W}, f)$ be a control system for which \Ass~\ref{ass:Wt} holds.  
       Given three constants $\tau \in \posreal{}$, $\StateQnt{} \in \posreal{}$, and $\IPQnt{} \in \posreal{}$, the abstract metric system induced by $\Sigma$ is defined by
       	\begin{equation}
       	 \ContSysA{}{\Sigma{}} = (\StateA{}, 
	                                \DiscSpace{U}{\IPQnt{}},\IPECurveA{},%\mathcal{U}_{\tau\StateQnt{},\IPQnt{}}, 
	                                \Wt,\IPICurveA{},%\mathcal{W}_{\tau\StateQnt{},\IPQnt{}},
	                                \deltaA{})
       	\end{equation}
    \st $\StateA{} = \DiscSpace{X}{\StateQnt{}}$, $\IPECurveA{}$ is defined over $\DiscSpace{U}{\IPQnt{}}$, as in \eqref{equ:def:mathcalU},
       \begin{equation*}\label{equ:piecewise constant W}
       		\IPICurveA{} := \set{\nu:[0,\tau]\rightarrow \Wt \mid \forall t,k\in [0,\tau]~.~\nu(t) = \nu(k)},
       \end{equation*}
       and for all $t_1,t_2\in [0,\tau)$, 
		 \begin{align*}
% 		  &\deltaA{}(t_1,x, \mu, \nu) = \deltaA{}(t_2,x, \mu, \nu) = x \text{ and }\\
		  &\deltaA{}(x, \mu, \nu) = \set{x'\in\StateAb{} \mid \norm{\StateTraj{}{x\mu\nu}(\tau) - x'} \leq \StateQnt{}}.
		 \end{align*}
% 		\end{align*}
% 	We denote the unique value of $\mu\in\IPECurveA{}$ and $\nu\in\IPICurveA{}$ over $[0,\tau]$ by  $u_\mu\in \DiscSpace{U}{\omega}$ and $w_\nu\in \Wt$, respectively.
	We equip $\StateA{}$ with the metric $\dist{x}{x'} \defeq \norm{x -x'}$. % and $\IPICurveA{}$ with the metric $\vmetric(\nu,\nu')=\norm{\nu-\nu'}_{\infty}=\norm{w_\nu - w_{\nu'}}$.
\end{definition}
% \KM{Not sure if we need the last inequality in \eqref{equ:Vtildeeps}.}

%\AKS{Sadegh,Kaushik, please check the following paragraph (until the theorem) carefully:}
Following the previous discussion, the obvious interpretation of $\IPICurveA{}$ in a network of symbolic abstract models is that $\IPICurveA{}$ actually collects the constant state trajectories of neighboring systems abstractions. 
% While we postpone a discussion of such networks to \Sec~\ref{} we want to emphasize that the metric system $\ContSysDT{\Sigma}$ is modeled with $\mathcal{W}_\tau$ containing continuous trajectory pieces, while $\ContSysDT{\Sigma}$ is modeled s.t.\ $\IPICurveA{}$ contains constant trajectory pieces. 
Therefore, given any $\nu\in\IPICurveDT{}$ and $\hat{\nu}\in\IPICurveA{}$ s.t. $\vmetric({\nu}(0),\hat{\nu}(0))\leq \DQnt{}$ holds, we see that the distance between $\hat{\nu}$ (constant signal) and $\nu$ (time-varying signal) potentially grows in the inter-sampling period. To ensure that $\ContSysDT{\Sigma}$ and $\ContSysA{}{\Sigma}$ are disturbance bisimilar, we have to make sure that the effect of this mismatch on the distance of the trajectories in both systems is small. This is obviously true if the effect of the disturbance on the dynamics of the underlying control system is small. This is formalized by the following assumption.

\begin{assumption}\label{ass:psi}
 Let $\Sigma$ be a control system with $\dISS$ Lyapunov function $V$ satisfying \eqref{eq:ISS_cond}. % \new{s.t. \eqref{equ:sensitivity bound} %holds for the constant $\psi^{(2)} \in \mathbb{R}_{>0}$.} 
	We assume there exists a constant $\psi > 0$ \st
	\begin{align}
		\frac{d}{dz}\sigma_{d}(z)\cdot\norm{\frac{d}{dt}\nu(t)} \leq \psi
	\end{align}
	holds for any $z\in[0,\norm{\tilde\varepsilon}]$, $t\in[0,\tau]$ and any disturbance $\nu\in\mathcal{W}_{\tau}$.
	%(\AKS{Kaushik: I've changed $\mathcal{W}$ to $\nu\in\mathcal{W}_{\tau}$ or do you mean $\nu:[0,\infty)\rightarrow W$? Also how is $w$ and $w'$ quantified?})
\end{assumption}

 Given \Ass~\ref{ass:psi} it is easy to show that the effect of the mismatch between $\nu$ and $\hat{\nu}$ on the state has a growth rate not greater than $\psi$, i.e.,
		\begin{align}\label{equ:proof:growthbound}
		 \frac{d\sigma_d\left(\norm{\hat{\nu}(t)-\nu(t)}\right)}{dt} \le \psi.
		\end{align}
From this observation we get the inequality
\begin{equation*}
	\sigma_d\left(\norm{\hat{\nu}(t)-\nu(t)}\right) - \sigma_d\left(\norm{\hat{\nu}(0)-\nu(0)}\right) \le \psi\cdot t,
\end{equation*}
%can  obtain a bound on $\sigma_d\left(\norm{\hat{\nu}(t)-\nu(t)}\right)$ for all $t\in[0,\tau]$ by assuming (without loss of generality) that this term becomes maximal at the end of the sampling interval (i.e. for $t=\tau$). Hence, by integrating both sides of \eqref{equ:proof:growthbound} for $t\in[0,\tau]$ we immediately get
%                \begin{align*}
%                 &\sigma_d\left(\norm{\hat{\nu}(\tau)-\nu(\tau)}\right) - \sigma_d\left(\norm{\hat{\nu}(0)-\nu(0)}\right) < \psi\cdot \tau\\
%		\Rightarrow &\sigma_d\left(\norm{\hat{\nu}(\tau)-\nu(\tau)}\right) < \sigma_d\left( \norm{\tilde{\varepsilon}} \right) + \psi\cdot \tau
%                \end{align*}
                and therefore 
                \begin{align}\label{equ:proof:growthbound_final}
                \forall t\in[0,\tau]~.~ \sigma_d\left(\norm{\hat{\nu}(t)-\nu(t)}\right)<\sigma_d\left( \norm{\tilde{\varepsilon}} \right) + \psi\cdot \tau.
                \end{align}

This observation will be used to prove our second main result which shows that  $\ContSysDT{\Sigma}$ and $\ContSysA{}{\Sigma}$ are disturbance bisimilar under \Ass~\ref{ass:psi}.

% Now we present the theorem which relates the two metric systems $\ContSysDT{\Sigma}$ and $\ContSysA{}{\Sigma}$.

\begin{theorem} \label{thm:sim-approx-single}
	Let $\Sigma$ be a control system with a $\dISS$ Lyapunov function $V$ satisfying \Ass~\ref{ass:psi} and property \eqref{eq:ISS_cond}.
	%, subjected to \new{slowly growing disturbance signals $\set{\nu\in \IPICurve{} \mid \norm{\frac{d\nu(t)}{dt}<\psi}}$ for some small $\psi \in \mathbb{R}_{>0}$}. Let $\Sigma$ admits a 
% 	admitting a $\dISS$ Lyapunov function 
% 	${\Lypv{}}$ witnessed by $\lambda$, $\alphalow{}$, $\alphahigh{}$, $\sigmau{}$, and $\sigmad{}$ s.t. 
% 	\Ass~\ref{ass:psi} holds
% 	% \new{s.t. \eqref{equ:sensitivity bound} %holds for the constant $\psi^{(2)} \in \mathbb{R}_{>0}$.} Let
% % 	the following hold:
% % 	\begin{align}
% % 		\frac{d\sigma_{d}}{d\norm{w-w'}}\cdot\norm{\frac{d\nu(t)}{dt}} \leq \psi
% % 	\end{align}
% % 	for a small $\psi > 0$ and any disturbance $\nu\in\mathcal{W}$.
% % 	Let 
%  and let $\gammafunc{}$ be a $\KInf$ function \st for any $\state{},\state{}',\state{}''\in\mathbb{R}^{n}$ it holds that
% 		\begin{align}
% 			\Lypv{}(\state{}',\state{}) - \Lypv{}(\state{}'',\state{}) \leq \gammafunc{}(\norm{\state{}'-\state{}''}).
% 		\end{align} 
	 Fix $\TimeQnt{}>0$ and $\Wt\subseteq W$ s.t. \Ass~\ref{ass:Wt} holds and let $\ContSysA{}{\Sigma{}}$ be the abstract metric system induced by $\Sigma$. 
% 	and suppose that there are parameters $\varepsilon\in\poszreal{}$,
% 	$\StateQnt{}\in\poszreal{}$ and $\IPQnt{} \in \posreal{}$ satisfying
If
	\begin{multline}
		\StateQnt{} \leq \textrm{min}\left\lbrace\gammafunc{}^{-1}\lambda^{-1}(1-e^{-\lambdaL{}\tau})\left[ \lambda\alphalow{}(\varepsilon) - \sigmau{}(\IPQnt{})\right.\right.\\
		\left.\left.- \sigmad{}(\norm{\DQnt{}}) - \psi\cdot\tau \right], (\alphahigh{})^{-1}\circ\alphalow{}(\varepsilon)\right\rbrace \label{eq:condition on StateQnt}
	\end{multline}

%			\begin{multline}
%				\StateQnt{i}\leq \min \Bigg\lbrace \gammafunc{i}^{-1}\Bigg( \alphalow{i}(\varepsilon_i)(1-e^{-\lambdaL{i}\TimeQnt{i}}) - 2\gammafunc{i}(2\IPQnt{i}) - \gammafunc{i}\bigg( \rho_{d,i}\bigg( \sum_{j\in\nbr(i)} \varepsilon_j \bigg) \bigg) \Bigg),\ConjFunc{\parinv{\alphahigh{i}}}{\alphalow{i}}(\varepsilon_{i})\Bigg\rbrace \label{eq:condition on StateQnt}
%			\end{multline}
	then the relation
	\begin{align}
			\eqclass{} = &\left\lbrace (q,\hat{q})\in\StateDT{}\times\StateA{} \ |\ \Lypv{}(q,\hat{q})\leq\alphalow{}(\varepsilon_{}) \right\rbrace \label{eqn:local simulation relation}
	\end{align}
	is a disturbance bisimulation with parameters $(\varepsilon,\tilde{\varepsilon})$ between $\ContSysDT{\ContSysSymb{}}$ and $\ContSysA{}{\ContSysSymb{},\Wt}$.
\end{theorem}

Before giving the proof of \Thm~\ref{thm:sim-approx-single} we want to point out that for a given $\psi$ we can select the time sampling parameter $\tau$ sufficiently small so that \eqref{eq:condition on StateQnt} is satisfied. More precisely,
	if $$\lambda\alphalow{}(\varepsilon) > \sigmau{}(\IPQnt{})+\sigmad{}(\norm{\DQnt{}})$$ we can select $\tau$ according to
$$\tau<\frac{1}{\psi}\left[ \lambda\alphalow{}(\varepsilon) - \sigmau{}(\IPQnt{})-\sigmad{}(\norm{\DQnt{}})\right]$$ which guarantees the existence of $\eta>0$ satisfying \eqref{eq:condition on StateQnt}.

\begin{proof}
	%\begin{enumerate}
		%\item[\textbf{(a)}]
		First note that \eqref{eq:condition on StateQnt} and \eqref{eqn:lyapunov condition 1} imply
	$\StateQnt{} \leq \ConjFunc{(\alphahigh{})^{-1}}{\alphalow{}(\varepsilon)} \leq (\alphahigh{})^{-1}(\alphahigh{}(\varepsilon)) = \varepsilon$ giving that $\StateQnt{}\leq\varepsilon$, hence ensuring that $\eqclass{}$ is surjective. Furthermore, observe that $\StateA{}\subset\StateDT{}$, hence the metric $d$ on $\StateDT{}$ is also a metric on $\StateA{}$. Now we prove the three parts of \Def~\ref{def:DisturbanceBisimulation} separately.\\
	\begin{inparaenum}[(a)]
		\item By definition of $\eqclass{}$ in \eqref{eqn:local simulation relation}, $(q,\hat{q})\in \eqclass{}$ implies $\Lypv{}(q,\hat{q})\leq\alphalow{}(\varepsilon_{})$. Using \eqref{eqn:lyapunov condition 1} this implies $\alphalow{}(\norm{q-\hat{q}})\leq\alphalow{}(\varepsilon_{})$ and it follows from $\alphalow{}$ being a $\kifunc$-function that $d(q,\hat{q}) = \norm{q-\hat{q}}\leq\varepsilon$. \\
% 		
%The statement follows from the definition of norm.
%   		(b) Given a pair $(\hat{q}, q)\in\eqclass{}$, for any $\mu\in\IPECurveDT{}$ s.t. for all $t\in\dom{\mu}$, $\mu(t) = u$, observe that there exists a $\hat{\mu} \in \IPECurveA{}$ s.t. for all $t\in\dom{\hat{\mu}}$, $\hat{\mu}(t) = \hat{u}$ and $\norm{\hat{u}-u}\leq \IPQnt{}$ holds.
		\item Given a pair $(\hat{q}, q)\in\eqclass{}$, for any $\mu\in\IPECurveDT{}$, observe that there exists a $\hat{\mu} \in \IPECurveA{}$ s.t.\ $\norm{u_{\mu}-u_{\hat{\mu}}}\leq \IPQnt{}$ holds.
                We furthermore pick $\hat{\nu}\in\IPICurveA{}$ and $\nu\in\IPICurveDT{}$ s.t. $\vmetric({\nu}(0),\hat{\nu}(0))\leq \DQnt{}$ holds. By applying transitions 
		$q\TransFuncDT{}{\mu,\nu}q'$ and $\hat{q}\TransFuncDT{}{\hat{\mu},\hat{\nu}}z$ we observe that there exists $\hat{q}'\in\StateA{}$ s.t.\ $\norm{\hat{q}'-z}\leq\eta$ and hence $\hat{q}\TransFuncA{}{\hat{\mu},\hat{\nu}}\hat{q}'$.
	   	Now consider Derivation~\eqref{equ:proof:derivation1} which uses \eqref{equ:proof:growthbound_final} obtained from \Ass~\ref{ass:psi}.
	   	\begin{table*}
%	   		\begin{align*}
%	   			\frac{dV}{dt}(x,x') &\leq -\lambda V(x,x') + \sigma_u(\norm{u-u'}) + \sigma_d(\norm{w-w'})\\
%	   			\Rightarrow e^{\lambda t}\left[ \frac{dV}{dt}(x,x') + \lambda V(x,x') \right] &\leq e^{\lambda t}\left[ \sigma_u\left( \norm{u-u'} \right) + \sigma_d\left( \norm{w-w'} \right) \right]\\
%	   			\Rightarrow \frac{d}{st}\left[ e^{\lambda t}V(x,x') \right] \leq 
%	   		\end{align*}
			\begin{align*}
				& \frac{dV}{dt}(x,x')\le -\lambda V(x,x')+\sigma_u(\|u-u'\|)+\sigma_d(\|w-w'\|)\\
				& \Rightarrow e^{\lambda t}\left[\frac{dV}{dt}(x,x')+\lambda V(x,x')\right]\le  e^{\lambda t}\left[\sigma_u(\|u-u'\|)+\sigma_d(\|w-w'\|)\right]\\
				& \Rightarrow \frac{d}{dt}\left[e^{\lambda t}V(x,x')\right]\le e^{\lambda t}\left[\sigma_u(\|u-u'\|)+\sigma_d(\|w-w'\|)\right]\\
				& \Rightarrow e^{\lambda \tau}V(z,q') - V(\hat q,q) \le \int_0^\tau e^{\lambda t} \sigma_u(\|\mu(t)-\hat\mu(t)\|)dt + \int_0^\tau e^{\lambda t} \sigma_d(\|\nu(t)-\hat\nu(t)\|)dt\\
				& \Rightarrow V(z,q')\le e^{-\lambda \tau} V(\hat q,q) + \int_0^\tau e^{-\lambda(\tau-t)} \sigma_u(\|\mu(t)-\hat\mu(t)\|)dt + \int_0^\tau e^{-\lambda(\tau-t)} \sigma_d(\|\nu(t)-\hat\nu(t)\|)dt\\
				& \Rightarrow V(z,q')\le e^{-\lambda \tau}\underline\alpha(\varepsilon)+\frac{1-e^{-\lambda\tau}}{\lambda}\left[\sup_{t\in[0,\tau]}\sigma_u(\|\mu(t)-\hat\mu(t)\|) + \sup_{t\in[0,\tau]}\sigma_d(\|\nu(t)-\hat\nu(t)\|)\right]\\
				& \Rightarrow V(z,q')\le e^{-\lambda \tau}\underline\alpha(\varepsilon)+\frac{1-e^{-\lambda\tau}}{\lambda}\left[\sigma_u(\|\mu(t)-\hat\mu(t)\|_\infty) + \sigma_d(\|\nu(t)-\hat\nu(t)\|_\infty)\right]\\
				%& 	\Rightarrow V(z,q')\le e^{-\lambda \tau}\underline\alpha(\varepsilon)+\frac{1-e^{-\lambda\tau}}{\lambda}\left[\sigma_u(\omega) + \sigma_d(\|\boldsymbol e(\nu,\hat\nu)\|)\right]\\
				& \Rightarrow V(z,q')\le e^{-\lambda \tau}\underline\alpha(\varepsilon)+\frac{1-e^{-\lambda\tau}}{\lambda}\left[\sigma_u(\omega) + \sigma_d(\|\tilde\varepsilon\|) + \psi\cdot\tau\right]\\
				& \Rightarrow e^{-\lambda \tau}\underline\alpha(\varepsilon)+\frac{1-e^{-\lambda\tau}}{\lambda}\left[\sigma_u(\omega) + \sigma_d(\|\tilde\varepsilon\|) + \psi\cdot\tau \right]+\gamma(\eta)\le\underline\alpha(\varepsilon) \numberthis \label{equ:proof:derivation1}
			\end{align*}
	   	\end{table*}
%	   	\begin{align}\allowdisplaybreaks
%	   		&\Lypv{}(\hat{q}',q') \label{equ:proof:derivation1}\\
%	   		&\leq \Lypv{}(z,q') + \gammafunc{}(\norm{z-\hat{q}'}) \notag\\
%	   		&\leq e^{-\lambdaL{}\tau{}}\Lypv{}(q,\hat{q}) + \frac{\sigmau{}(\norm{u_{\hat{\mu}}-u_{\mu}})}{\lambdaL{}} + \frac{\sigmad{}(\norm{w_{\hat{\nu}}-w_{\nu}})}{\lambdaL{}} + \gammafunc{}(\StateQnt{}) \notag\\
%	   		&\leq e^{-\lambdaL{}\tau{}}\alphalow{}(\varepsilon) + \frac{\sigmau{}(\IPQnt{})}{\lambdaL{}} + \frac{\sigmad{}(\norm{\DQnt{}})}{\lambdaL{}} + \gammafunc{}(\StateQnt{})
%	   		\leq \alphalow{}(\varepsilon)\notag
%	   	\end{align}
	   	Hence by \Eqn~\eqref{eqn:local simulation relation}, $(q',\hat{q}')\in\eqclass{}$.\\
	   	%\item[\textbf{(c)}] 
%   		(c) Given a pair $(\hat{q}, q)\in\eqclass{}$, for any $\hat{\mu}\in\IPECurveA{}$ s.t. for all $t\in\dom{\hat{\mu}}$, $\hat{\mu}(t) = \hat{u}$, observe that we can choose $\mu \in \IPECurveDT{}$ s.t. for all $t\in\dom{\mu}$, $\mu(t) = u = \hat{u}$.
		\item Given a pair $(\hat{q}, q)\in\eqclass{}$, for any $\hat{\mu}\in\IPECurveA{}$, observe that we can choose $\mu \in \IPECurveDT{}$ s.t. $\mu = \hat{\mu}$, i.e., $\norm{u_{\hat{\mu}}-u_\mu} = 0$.
   		Given any $\nu\in\IPICurveDT{}$ and $\hat{\nu}\in\IPICurveA{}$ s.t.\ $\vmetric(\nu(0),\hat{\nu}(0))\leq \DQnt{}$, we get $q\TransFuncDT{}{\mu,\nu}q'$ and $\hat{q}\TransFuncDT{}{\hat{\mu} = \mu,\hat{\nu}}z$. Now observe that there exists $\hat{q}'\in\StateA{}$ s.t.\ $\norm{\hat{q}'-z}\leq\eta$ and hence $\hat{q}\TransFuncA{}{\hat{\mu},\hat{\nu}}\hat{q}'$.
%    		\begin{align*}
% 	   		&q\TransFuncDT{}{\mu,\nu}q'\\
% 	   		&\hat{q}\TransFuncDT{}{\hat{\mu} = \mu,\hat{\nu}}z \Rightarrow
% 	   			\exists\hat{q}'\in\StateA{}.\norm{\hat{q}'-z}\leq\eta \Rightarrow 
% 	   					
% 	   	\end{align*}
		With a very similar derivation as in \eqref{equ:proof:derivation1} it follows from \Eqn~\eqref{eqn:local simulation relation} that $(q',\hat{q}')\in\eqclass{}$.
% 	   	Now consider the following derivation:
% 	   	\begin{align*}
% 	   		\Lypv{}(\hat{q}',q') &\leq \Lypv{}(z,q') + \gammafunc{}(\norm{z-\hat{q}'})\\
% 	   		&\leq e^{-\lambdaL{}\tau{}}\Lypv{}(q,\hat{q}) + \frac{\sigmad{}(\norm{\hat{w}_{\hat{\nu}}-w_{\nu}})}{\lambdaL{}} + \gammafunc{}(\StateQnt{})\\
% 	   		&\leq e^{-\lambdaL{}\tau{}}\alphalow{}(\varepsilon) + \frac{\sigmad{}(\DQnt{})}{\lambdaL{}} + \gammafunc{}(\StateQnt{})\\
% 	   		&\leq \alphalow{}(\varepsilon)
% 	   	\end{align*}
% 	   	Hence by \Eqn~\eqref{eqn:local simulation relation}, $(\hat{q}',q')\in\eqclass{}$. This concludes the proof.
	%\end{enumerate}
	\end{inparaenum}
\end{proof}
% \KM{Inconsistency in domain of $\vmetric$.}
% \AKS{I don't see any difference between \eqref{eqn b} and \eqref{eqn c}. Should there be any?}

% \new{
\begin{remark}\label{rem:CompactStateset}
 As we have defined control systems $\Sigma$ w.r.t.\ the Euclidean spaces $X=\real{n}$, $U=\real{m}$ and $W=\real{p}$, the abstract metric system $\ContSysA{}{\Sigma{}}$ only becomes finite (and therefore a symbolic abstraction of the control system $\Sigma$) if we restrict its construction to compact subsets $X'\subset X$ and $U'\subset U$ of the state and input spaces, s.t.\ $X'$ and $U'$ are finite unions of hyper-rectangles with radius $\eta$ and $\omega$, respectively.
 In the control system networks we consider, such a restriction also implies that the disturbance space $W$ becomes compact. 
 However, when synthesizing controllers for the abstract metric system, it needs to be ensured that the closed loop dynamics do not leave the selected compact state set $X'$. This can be done by treating $X'$ as an additional safety constrain during synthesis. We will illustrate this approach in our case study presented in \Sec~\ref{sec:Example}.
\end{remark}
% }

%\AKS{Kaushik/Sadegh: I think the first part of the proof needs to be modified to show that given a small $\psi$  \eqref{eq:condition on StateQnt} can always be fulfilled. Or stated differently, is there a condition on $\psi$ under which the parameters $\eta$, $\omega$, $\epsilon$ and $\tilde{\varepsilon}$ exist s.t. \eqref{eq:condition on StateQnt} holds?}

\section{Compositional Abstraction} \label{sec:composition}

We now extend the abstraction procedure presented in the previous section to compositions of control systems.

\subsection{Networks of Control Systems}\label{sec:NetworkControlSystems}
We define networks of control systems in direct analogy to \Sec~\ref{sec:Comp:Metric}. 
Let $\Sigma_i = (X_i, U_i, \mathcal{U}_i, W_i, \mathcal{W}_i, f_i)$,
for $i\in I$, be a control system.
We say that the set of control systems $\set{\Sigma_i}_{i\in I}$ are \emph{compatible for composition} w.r.t.\ the interconnection relation $\mathcal{I}$, if 
for each $i\in I$, we have
$W_i = \prod_{j\in \nbr_\mathcal{I}(i)} {X_j}$,
% , i.e., the disturbance input space of $\Sigma_i$ is the same as the cartesian product of the state spaces of all the neighbors in $\nbr_{\mathcal{I}}(i)$. \new{By slightly abusing notation we write $w_i=\prod_{j\in \nbr_\mathcal{I}(i)} \set{x_j}$ as a short form of $\set{w_i}=\prod_{j\in \nbr_\mathcal{I}(i)} \set{x_j}$. We extend this notation to all tuples of elements from the sets $X_i$, $U_i$ and $W_i$.}
divided in coupling and external disturbances $W_i^c$ and $W_i^e$, respectively, as defined in \Sec~\ref{sec:Comp:Metric}.
% 
% As $I'$ is a subset of all systems in the network, we divide the set of disturbances $W_i$ for any $i\in I'$ into the sets of coupling and external disturbances, defined by
% $W_i^c = \prod_{j\in \nbr_{\mathcal{I}'}(i)} {X_j}$ and $W_i^e = \prod_{j\in\nbr_{\IntCon\setminus\IntCon'}(i)} {X_j}$, respectively. 

If $\set{\Sigma_i}_{i\in I}$ are compatible, we define the \emph{composition} of any subset $I'\subseteq I$ of systems as the control system
$\llbracket \Sigma_i\rrbracket_{i\in I'}= (X, U, \mathcal{U}, W, \mathcal{W}, f)$ where $X$, $U$ and $W$ are defined as in \Sec~\ref{sec:Comp:Metric}.
% $X = \prod_{i\in I'} {X_i}$,
% $U = \prod_{i\in I'} {U_i}$, and
% $W = \prod_{j\in \nbr_{\mathcal{I}}(I')} {X_j}$.
% 
Furthermore, $\mathcal{U}$ and $\mathcal{W}$ are defined as the sets of functions $\mu : \poszreal{} \rightarrow U$ and $\nu : \poszreal{} \rightarrow W$,
such that the projection $\mu_i$ of $\mu$ on to $U_i$ (written $\mu_i=\mu|_{U_i}$) belongs to $\mathcal{U}_i$, and the projection $\nu_i^e$ of $\nu$ on to $W_i^e$ belongs to $\mathcal{W}_i^e$.
The composed transition function is then defined as $f(\prod_{i\in I'} \set{x_i}, \prod_{i\in I'} \set{u_i}, \prod_{i\in I'} \set{w_i^e}) = \prod_{i\in I'} \set{f_i(x_i, u_i, w_i^c \times w_i^e)}$, where $w_i^c = \prod_{j\in \nbr_{\mathcal{I}'}(i)} \set{x_j}$. 
If $I'=I$, then $\Sigma$ is undisturbed, modeled by $W:=\set{0}$.
% $W$ and $\mathcal{W}$ are both $\emptyset$, and the transition function simplifies to $f: X \times \mathcal{U} \rightarrow X$. 
It is easy to see that $\llbracket \Sigma_i\rrbracket_{i\in I'}$ is again a control system.

\subsection{Assumptions on Interconnecting Disturbances}\label{sec:CA:Assumptions}
Given $I$ and $I' \subseteq I$, consider a set of compatible control systems $\set{\Sigma_i}_{i\in I}$, the subset composition $\llbracket \Sigma_i\rrbracket_{i\in I'}= (X, U, \mathcal{U}, W, \mathcal{W}, f)$ and 
a global time-sampling parameter $\tau$. 
Then we can apply \Def~\ref{def:control-sys-disctime} and \Def~\ref{def:abstract system} to each control system $\Sigma_i$ to construct the corresponding 
metric systems $\ContSysDT{\Sigma_i}$ and $\ContSysAi{\Sigma_i}$. To be able to do that, we need to define $\Wt_i$ for all $i\in I$ s.t. Ass.~\ref{ass:Wt} holds.

\begin{lemma}\label{lem:Wt}
 Let $\set{\Sigma_i}_{i\in I}$ be a set of compatible control systems and $\set{\ContSysAi{\Sigma_i}}_{i\in I}$ the set of abstract metric systems, each induced by $\Sigma_i$, respectively, where
 \begin{equation}\label{equ:Wti}
  \Wt_i=\prod_{j \in \nbr_{\mathcal{I}}(i)} \StateA{j}.
 \end{equation}
If all local quantization parameters $\set{\varepsilon_i}_{i\in I}$ and $\set{\StateQnt{i}}_{i\in I}$ fulfill \eqref{equ:tildevarepsilon 1} and $\eta_i\leq\varepsilon_i$ for all $i\in I$ then \eqref{equ:Vtildeeps} in Ass.~\ref{ass:Wt} holds for every $i\in I$ w.r.t.\ the metric defined in \eqref{equ:defvmetric_composed}.
\end{lemma}

\begin{proof}
Pick any $i\in I$, $w_i\in W_i$ and observe that $w_i=\prod_{j \in \nbr_{\mathcal{I}}(i)} \set{x_{j}}$. By the choice of $\StateA{j}$ as $\DiscSpace{X_j}{\StateQnt{j}}$ we furthermore know that for any $x_j$ there exists $\hat{x}_j$ \st $\norm{x_j-\hat{x}_j}\leq\eta_j\leq\varepsilon_j$. Now recall that  $\Wt_i=\prod_{j \in \nbr_{\mathcal{I}}(i)} \StateA{j}=\prod_{j \in \nbr_{\mathcal{I}}(i)} \DiscSpace{X_j}{\StateQnt{j}}$. 
Using the definition of $\tilde{\varepsilon}_i$ in \eqref{equ:tildevarepsilon 1} and $\vmetric$ in \eqref{equ:defvmetric_composed} we therefore know that for any $w_i\in W_{i}$ there exists $\tilde{w}_i\in \Wt_i$ s.t. $\vmetric(w_i,\wt_i)=\prod_{j \in \nbr_{\mathcal{I}}(i)} \set{\norm{x_j - \hat{x}_j}} \leq\prod_{j \in \nbr_{\mathcal{I}}(i)} \set{\varepsilon_j}=\tilde{\varepsilon}_i$. %Using the triangular inequality furthermore yields $\norm{w_i-\wt_i}\leq\sum_{j \in \nbr_{\mathcal{I}}(i)} \set{\norm{x_j-\hat{x}_j}}=\norm{\vmetric(w_i,\wt_i)}$. 
Furthermore, $\norm{w_i-\wt_i} = \norm{\prod_{j\in\nbr_{\mathcal{I}}(i)}\set{x_j - \hat{x}_j}} = \norm{\prod_{j\in\nbr_{\mathcal{I}}(i)}\set{\norm{x_j - \hat{x}_j}}} = \norm{\vmetric(w_i,\wt_i)}$. 
%Hence \eqref{equ:Vtildeeps} holds for all $\ContSysAi{\Sigma_i}$ and $\tilde{\varepsilon}_i$. 
\end{proof}

Given \Lem~\ref{lem:Wt}, it immediately follows that the sets $\set{\ContSysDT{\Sigma_i}}_{i\in I'}$ and $\set{\ContSysAi{\Sigma_i}}_{i\in I'}$ of metric systems are again compatible. 

To prove \Thm~\ref{thm:sim-approx-single} we have additionally used \Ass~\ref{ass:psi} which essentially bounds the effect of the disturbances on the state evolution. Given the particular choice of disturbances in the network as state trajectories of neighboring systems, we can replace \Ass~\ref{ass:psi} with the following assumption.

\begin{assumption}\label{ass:psi_i}
 Let $\set{\Sigma_i}_{i\in I}$ be a set of compatible control systems, each admitting a $\dISS$ Lyapunov function 
	${\Lypv{i}}$ witnessed by $\lambda_i$, $\alphalow{i}$, $\alphahigh{i}$, $\sigmau{i}$, and $\sigmad{i}$. % \new{s.t. \eqref{equ:sensitivity bound} %holds for the constant $\psi^{(2)} \in \mathbb{R}_{>0}$.} Let
	Then there exist constants $\psi_i > 0$ s.t.\
	\begin{multline}
		\forall x_i\in X_i.\forall u_i\in U_i.\forall w_i\in W_i.z_i\in[0,\norm{\tilde\varepsilon_i}].\\ 
	 									\frac{d}{dz_i}\sigma_{d,j}(z_i)\cdot\norm{\prod_{i\in\nbr_\IntCon{}(j)} \set{f_i(x_i,u_i,w_i)}} \leq \psi_{j} \label{eq:weak interconnection}
	\end{multline}
	holds.% for any disturbance $\nu_i\in\mathcal{W_i}$.
\end{assumption}
If \Ass~\ref{ass:psi_i} holds for a network of compatible control systems $\set{\Sigma_i}_{i\in I}$, we call this network \emph{weakly interconnected}.

We have the following obvious lemma connection \Ass~\ref{ass:psi} with \Ass~\ref{ass:psi_i}.

\begin{lemma}\label{lem:psi}
 Let $\set{\Sigma_i}_{i\in I}$ be a set of compatible control systems, each admitting a $\dISS$ Lyapunov function 
	${\Lypv{i}}$ witnessed by $\lambda_i$, $\alphalow{i}$, $\alphahigh{i}$, $\sigmau{i}$, and $\sigmad{i}$. Then \Ass~\ref{ass:psi} holds for all $\Sigma_i$ iff \Ass~\ref{ass:psi_i} holds for the network $\set{\Sigma_i}_{i\in I}$.
\end{lemma}

\begin{proof}
%\AKS{Kaushik/Sadegh this needs to be a bit more detailed}
Follows directly from the fact that
\begin{align*}
 \frac{d\nu_i(t)}{dt}
 = \prod_{j\in\nbr_{\mathcal{I}}(i)} \left\lbrace\frac{dx_j(t)}{dt}\right\rbrace
 = \prod_{j\in \nbr_{\mathcal{I}}(i) } \set{f_j(x_j,u_j,w_j)},
\end{align*}
for each subsystem $i\in N$.
%holds for all $i$.
%  We obviously have for all $j\in \nbr_\mathcal{I}(i)$ that
%  \begin{align*}
%   \frac{d\sigma_{d,i}}{d\norm{w_i-w_i'}}\cdot\norm{f_j(x_j,u_j,w_j)} = \frac{d\sigma_{d,i}}{d\norm{w_i-w_i'}}\cdot\norm{\frac{d\nu_i(t)}{dt}} \leq \psi_{i}
%  \end{align*}
%  what proves the statement
\end{proof}

\begin{remark}
%\AKS{Sadegh: Can you manybe reformulate this? Also, this should also be included somewhere in the introduction maybe?}
Here we want to draw the attention of the reader to the importance of \Ass~\ref{ass:psi_i}.
The state trajectories of the network of continuous systems evolve continuously, whereas the abstract systems update their states only at the sampling instances. Hence in order to balance the accumulated errors during the inter-sampling periods, the abstraction process needs to be somewhat pessimistic. This problem does not appear in \cite{tazaki2008bisimilar}, where the dynamical systems are assumed to be sampled time systems and update their states together with their corresponding abstract systems.
%
% Recall that we have shown in \Sec~\ref{sec:MonolithicAbstraction} that using \Ass~\ref{ass:psi} we can bound the mismatch of disturbance trajectories in the inter-sampling periods. Relating this mismatch to state trajectories as in \Ass~\ref{ass:psi_i} gives us a bound on the mismatch between state trajectories of all systems as a by-product. I.e., given weakly interconnected systems we can give guarantees on the behaviour of state trajectories in the inter-sampling periods. This is usually not true for controller synthesis methods using $\varepsilon$-approximate bisimilar abstractions. Exceptions are \cite{?} %Neycmys paper, Antoins paper
% Our method can therefore be used to prove stronger guarantees for refined controllers applied to the original system.
% \AKS{What else do we want to say here? I'm not happy with this remark yet. Maybe be more precise?}
% \AKS{we also have to add the reference suggestion of rev. 1 here. Our work is more powerful, as their method only works for discrete time control systems.}
% \AKS{Also we might consider discussing in \Sec~\ref{sec:control} that with this we give stronger guarantees for the closed loop system....}
\end{remark}

% 	\new{Additionally, we say that a set of control systems $\set{\Sigma_i}_{i\in I}$ are \emph{weakly interconnected} w.r.t. the interconnection relation $\IntCon$ iff the following holds:
% 	 %two conditions are satisfied:
% 	 For all $j\in I$ the following should hold:
% 	 \begin{multline}
% 	 	\forall x_i\in X_i.\forall u_i\in U_i.\forall w_i\in W_i.\\ 
% 	 									\frac{d\sigma_{d,j}}{d\norm{w_j-w_j'}}\cdot\norm{\prod_{i\in\nbr_\IntCon{}(j)} f_i(x_i,u_i,w_i)} \leq \psi_{j} \label{eq:weak interconnection}
% 	 \end{multline}
% 	 for some small $\psi_{ij} > 0$.
% %
% 	A set of interconnected control systems, which are not weakly interconnected, are strongly interconnected.
% 	We do not discuss at this point, what does it mean to have \emph{small} values for $\psi_{j}$. %$\psi^{(1)}_i$ and $\psi^{(2)}_i$. 
% 	But, as it will be clear in the subsequent sections that, the bigger the $\psi_{j}$, %$\psi^{(1)}_i$ and $\psi^{(2)}_i$, 
% 	the more conservative will be the control.}
% 	\KM{Try to make a connection with bounded vector field and weak coupling in MIMO systems here.}

\subsection{Simultaneous Approximation}

Recall from \Sec~\ref{sec:MonolithicAbstraction} that under \Ass~\ref{ass:Wt} and \Ass~\ref{ass:psi} we have shown by \Thm~\ref{thm:sim-approx-single} that for any control system $\Sigma$ its corresponding metric systems $\ContSysDT{\Sigma}$ and $\ContSysAi{\Sigma}$ constructed via \Def~\ref{def:control-sys-disctime} and \Def~\ref{def:abstract system} are disturbance bisimilar if \eqref{eq:condition on StateQnt} holds for the discretization parameters involved in the abstraction process. While in the monolithic case we can freely choose all these parameters in a way that \eqref{eq:condition on StateQnt} holds, this is no longer true if we abstract a network of control systems $\set{\Sigma_i}_{i\in I}$ simultaneously. Namely, the $\tilde{\varepsilon}_i$-s depend on the precision parameters of the neighboring systems $\varepsilon_j$-s ($j\in\nbr_{\mathcal{I}}(i)$), and the sampling parameter $\tau$ has to be the same for all subsystems.

 By furthermore resolving \Ass~\ref{ass:Wt} using \Lem~\ref{lem:Wt} we see that \Thm~\ref{thm:sim-approx-single} can be generalized to networks of \emph{weakly interconnected} control systems if \eqref{eq:condition on StateQnt} and \eqref{equ:defvmetric_composed} can be fulfilled simultaneously by all discretization parameter sets involved in the abstraction process. This intuition is formalized by the following corollary, which is a direct consequence of \Thm~\ref{thm:sim-approx-single}, \Lem~\ref{lem:Wt} and \Lem~\ref{lem:psi}.

\begin{corollary}\label{cor:sim-appx}
 Let $\set{\Sigma_i}_{i\in I}$ be a set of compatible and weakly interconnected control systems that have $\dISS$ Lyapunov functions $V_i$ satisfying \eqref{eq:ISS_cond}. 
Let $\set{\ContSysDT{\Sigma_i}}_{i\in I}$ be the set of discrete-time metric systems 
induced by $\set{\Sigma_i}_{i\in I}$ and let $\set{\ContSysAi{\Sigma_i}}_{i\in I}$ be the set of abstract metric systems
  induced by $\set{\Sigma_i}_{i\in I}$ with $\Wt_i$ as in \eqref{equ:Wti}.
 % by locally applying \Def~\ref{def:control-sys-disctime} to every subsystem. 
If all local quantization parameters $\set{\varepsilon_i}_{i\in I}$, $\set{\tilde{\varepsilon}_i}_{i\in I}$, $\set{\IPQnt{i}}_{i\in I}$ and $\set{\StateQnt{i}}_{i\in I}$ simultaneously fulfill \eqref{equ:tildevarepsilon 1} and
	\begin{align*}
		0 < \StateQnt{i} \leq \textrm{min}\left\lbrace\gammafunc{i}^{-1}\lambda_i^{-1}(1-e^{-\lambdaL{i}\tau})\left[ \lambda_i\alphalow{i}(\varepsilon_i) - \sigmau{i}(\IPQnt{i})\right.\right.\\
		\left.\left.- \sigmad{i}(\norm{\DQnt{i}}) - \psi_i\cdot\tau \right], (\alphahigh{i})^{-1}\circ\alphalow{i}(\varepsilon_i)\right\rbrace, \numberthis \label{eq:condition on StateQnt i}
	\end{align*}
then the relation
\begin{align*}
			\eqclass{i} = &\left\lbrace (q_i,\hat{q}_i)\in\StateDT{i}\times\StateA{i} \ |\ \Lypv{i}(q_i,\hat{q}_i)\leq\alphalow{i}(\varepsilon_{i}) \right\rbrace %\label{eqn:local simulation relation}
	\end{align*}
	is a disturbance bisimulation relation with parameters $(\varepsilon_i,\tilde{\varepsilon}_i)$ between
	  $\ContSysDT{\ContSysSymb{i}}$ and $\ContSysA{i}{\ContSysSymb{i}}$.	
\end{corollary}

Given this result on simultaneous approximation we still need to answer when \eqref{equ:tildevarepsilon 1} and \eqref{eq:condition on StateQnt i} can be fulfilled simultaneously in a network of weakly interconnected control systems. This brings us to our third main result.

\begin{theorem}\label{thm:cond:sim-appx}
 Let $\set{\Sigma_i}_{i\in I}$ be a set of compatible and weakly interconnected control systems that have $\dISS$ Lyapunov functions $V_i$ satisfying \eqref{eq:ISS_cond}.
 Suppose for all $i\in I'$ there exist $\KInf$ functions $\vartheta_i$ and constants $c_{i,\alpha}\in \mathbb{R}_{>0}$ and $c_{i,\sigma}\in\mathbb{R}_{>0}$ s.t.
        \begin{enumerate}
		\item $\forall r\in\mathbb{R}_{\geq 0}$, $\alphalow{i}\left({r}/{\max\set{1,|\nbr_{\mathcal{I}'}(i)|}}\right)\geq c_{i,\alpha}\vartheta_i(r)$, and \label{assume a}
		\item $\forall r\in\mathbb{R}_{\geq 0}$, $\nbr_{\mathcal{I}'}(i)\neq \emptyset \Rightarrow \sigmad{i}(r)\leq c_{i,\sigma}\vartheta_i(r)$ \label{assume b}
	\end{enumerate}
 and there exists $s\in \mathbb{R}_{>0}^{|I'|}$ s.t. 
	\begin{align}\label{equ:small gain kind}	
	(-A+B)s < 0
	\end{align} holds,
	where $A\in \mathbb{R}^{N\times N}$ is a diagonal matrix with $A(i,i) = \lambda_i\cdot c_{i,\alpha}$ and $B \in \mathbb{R}^{N\times N}$ is s.t.\ $B(i,j) = c_{i,\sigma}$ if $j\in \nbr_{\mathcal{I}'}(i)$ and $B(i,j) = 0$ otherwise. Note that we assume that the network does not have any self loop, hence $B(i,i) = 0$ for all $i$.
	
	Then there exist a global time sampling parameter $\tau$ and sets of local quantization parameters $\set{\varepsilon_i}_{i\in I}$, $\set{\tilde{\varepsilon}_i}_{i\in I}$, $\set{\IPQnt{i}}_{i\in I}$ and $\set{\StateQnt{i}}_{i\in I}$ s.t. \eqref{equ:tildevarepsilon 1} and \eqref{eq:condition on StateQnt i} can be satisfied simultaneously.
\end{theorem}

\begin{proof}
%\AKS{I have not checked this.}
%  First note that under the condition that \eqref{equ:small gain kind} holds, \eqref{eq:condition on StateQnt i} is \emph{satisfiable for all $i\in I'$} for some appropriate value of $\eta_i$, $\varepsilon_i$, $\tau$ and $\omega_i$. %It is not difficult to see that given any system $i$, $|\nbr_{\mathcal{I}'}(i)| = 0$ implies that one can pick $\omega_i$ and $\tau$ sufficiently small in order to get positive upper bound on $\eta_i$. Consider the case when for some $i\in I'$, $|\nbr_{\mathcal{I}'}(i)| > 0$ and 
	Let $s = [\vartheta_i(\varepsilon'_1) \ \ldots \ \vartheta_i(\varepsilon'_{|I'|})]^T$ be one satisfying assignment of \eqref{equ:small gain kind} for some $\set{\varepsilon'_i}_{i\in I'}$, and consider the following derivation: 
	\begin{align*}	
	&(-A+B)s < 0 
	\Rightarrow (A-B)s > 0 \\
	&\Rightarrow \lambda_ic_{i,\alpha}\vartheta_i(\varepsilon'_i) - \sum_{j\in \nbr_{\mathcal{I}}(i)} c_{j,\sigma}\vartheta_j(\varepsilon'_j) > 0 \\
	&\Rightarrow \lambda_i\alphalow{i}(\varepsilon'_i/ |\nbr_{\mathcal{I}'}(i)|) - \sum_{j\in \nbr_{\mathcal{I}}(i)} \sigmad{j}(\varepsilon'_j) > 0 \\
	&\Rightarrow \lambda_i\alphalow{i}(\varepsilon'_i/ |\nbr_{\mathcal{I}'}(i)|) - \sigmad{j}\left(\sum_{j\in \nbr_{\mathcal{I}}(i)}(\varepsilon'_j/ |\nbr_{\mathcal{I}'}(i)|)\right) > 0. 
	\end{align*}
	%\Rightarrow \lambda_i\alphalow{i}(\varepsilon_i/ |\nbr_{\mathcal{I}'}(i)|) - \sigmad{j}\left( \norm{\tilde{\varepsilon}}/ |\nbr_{\mathcal{I}'}(i)|\right) > 0$. 
	Then by picking $\varepsilon_i = \varepsilon'_i/|\nbr_\mathcal{I}(i)|$, one can ensure that $\lambda_i\alphalow{i}(\varepsilon_i) - \sigmad{j}\left( \norm{\tilde{\varepsilon}}\right) > 0$ holds for all $i\in I'$. As a consequence, one can find suitable $\eta_i, \tau, \omega_i \in \mathbb{R}_{>0}$ which satisfy \eqref{eq:condition on StateQnt i}.
\end{proof}

%\AKS{Is there some intution behind \Thm~\ref{thm:cond:sim-appx} that we can include here? Is there a special case that gives more insight in these conditions besides the one covered in the remark below? We need to add more text here..}

\Thm~\ref{thm:cond:sim-appx} gives a generalized version of small-gain like conditions (see \cite[\Page 217]{khalil1996noninear}) for the existence of a solution of the simultaneous approximation problem. If \Inq~\ref{equ:small gain kind} is unsatisfiable for all $s$, then there exists $i\in N$ s.t. the $i$-th row of $B$ dominates the $i$-th row of $A$. Intuitively this means that there exists at least one system in the network which is too sensitive to its disturbances. In other words, the system's own dynamics are too weak to counteract its disturbances. This makes the problem of simultaneous approximation infeasible.

The conditions in \Thm~\ref{thm:cond:sim-appx} can be simplified significantly if the dynamics of the control systems $\Sigma_i$ satisfy some additional properties: 
%
%\AKS{todo rewrite the following}. 
\begin{inparaenum}[(a)]
\item If for all $i$, $\sigmad{i}$ satisfies the triangular inequality, then Condition~\eqref{assume a} in \Thm~\ref{thm:cond:sim-appx} can be replaced by the weaker condition $\forall r\in\mathbb{R}_{\geq 0}$, $\alphalow{i}\left({r}\right)\geq c_{i,\alpha}\vartheta_i(r)$. 
\item If the control systems $\set{\Sigma_i}_{i\in I'}$ are linear then $\alphalow{i}$ and $\sigmad{i}$ are constants, and in that case one can replace $c_{i,\alpha}$ and $c_{i,\sigma}$ by the constants $\alphalow{i}$ and $\sigmad{i}$ respectively, and use the linear function $\vartheta_i(r) = r$.
\end{inparaenum} %Note that for network of linear control systems, one can simply pick $c_{i,\alpha} = \alpha_i$, $c_{i,\sigma} = \sigma_{d,i}$ and $\vartheta(r) = r$ for all $i\in I'$ and $r\in \mathbb{R}_{\geq 0}$.

\begin{remark}\label{rem:small gain}
	Note that \eqref{equ:small gain kind} holds if $\lambda_{max}(A^{-1}B) < 1$ \cite[\Lem~3.1]{dashkovskiy2011small}, where $\lambda_{max}(\cdot)$ represents the maximum eigenvalue. For a two system network where both systems are connected to each other s.t.\ a cycle is formed, the eigenvalues of the matrix product $A^{-1}B$ are given by $\pm \sqrt{\frac{c_{1,\sigma}c_{2,\sigma}}{\lambda_1\lambda_2 c_{1,\alpha}c_{2,\alpha}}}$. Then by setting $c_{i,\sigma} = d_i$, $\lambda_i = l_i$ and $c_{i,\alpha} = 1$ for $i \in \set{1,2}$, we obtain the inequality $\frac{d_1d_2}{l_1l_2} < 1$ as a sufficient condition for \eqref{equ:small gain kind}, which is a small gain type condition similar to the ones presented in \cite[\Thm~2]{girard2013composition} and \cite{rungger2015compositional} in the context of composition of bisimulation and simulation functions of two interconnected subsystems, respectively.
\end{remark}

We will illustrate \Cor~\ref{cor:sim-appx} and \Thm~\ref{thm:cond:sim-appx} by an example in \Sec~\ref{sec:Example}.
%\AKS{todo:go over the example an check its compatibility to the new formalization. But maybe we replace this example anyway.}

%\input{running_ex_p1}

 \subsection{Composition of Approximations}
 
We have discussed in \Sec~\ref{sec:CA:Assumptions} that the sets $\set{\ContSysDT{\Sigma_i}}_{i\in I'}$ and $\set{\ContSysAi{\Sigma_i}}_{i\in I'}$ of metric systems are compatible.
% Consider a set of compatible control systems $\set{\Sigma_i}_{i\in I}$, a subset composition $\llbracket \Sigma_i\rrbracket_{i\in I'}= (X, U, \mathcal{U}, W, \mathcal{W}, f)$,
% and a global time-sampling parameter $\tau$. 
% Then we can apply \Def~\ref{def:control-sys-disctime} and \Def~\ref{def:abstract system} to each control system $\Sigma_i$ to construct the corresponding 
% metric systems
% $\ContSysDT{\Sigma_i}=(X_i, U_i, \mathcal{U}_{\tau,i}, W_i, \mathcal{W}_{\tau,i}, \delta_{\tau,i})$ and $\ContSysAi{\Sigma_i} = (\StateA{i}, \DiscSpace{U_i}{\omega_i}, \IPECurveA{i},\DiscSpace{W_i}{\tilde{\varepsilon}_i},\IPICurveA{i},\deltaA{i})$. Now it immediately follows that for any $I'\subseteq I$, the sets $\set{\ContSysDT{\Sigma_i}}_{i\in I'}$ and $\set{\ContSysAi{\Sigma_i}}_{i\in I'}$ of metric systems are again compatible.
Therefore, combining the results from \Thm~\ref{thm:composition of ADB MS}, \Cor~\ref{cor:sim-appx} and \Thm~\ref{thm:cond:sim-appx} leads to the following corollary.

\begin{corollary}\label{cor:comp-appx}
		Given the preliminaries of \Thm~\ref{thm:cond:sim-appx} and $I'\subseteq I$, let $\llbracket\ContSysDT{\Sigma_i}\rrbracket_{i\in I'}$ and $\llbracket\ContSysAi{\Sigma_i}\rrbracket_{i\in I'}$
		be systems composed from the sets $\set{\ContSysDT{\Sigma_i}}_{i\in I}$ and $\set{\ContSysAi{\Sigma_i}}_{i\in I}$, respectively. Then the relation
			\begin{align}\label{equ:eqclass}
				\eqclass{} = &\lbrace (\transpose{ q^T_{1}\delimit\ldots\delimit q^T_{|I'|}}, \transpose{\hat{q}^T_{1}\delimit\ldots\delimit\hat{q}^T_{|I'|}})\in\StateDT{}\times\StateAb{} \ |\notag\\
				&\quad(q_i,\hat{q}_i)\in\eqclass{i}, \forall i\in I'
				)\rbrace
			\end{align} 
		 is a disturbance bisimulation relation between $\llbracket\ContSysDT{\Sigma_i}\rrbracket_{i\in I'}$ and $\llbracket\ContSysAi{\Sigma_i}\rrbracket_{i\in I'}$ with parameters 
		 		\begin{align*}%\label{equ:pickprecision}
				\textstyle\varepsilon{} \textstyle= \norm{\prod_{i\in I'}\set{\varepsilon_i}} \text{ and }~ %\inv{\alphalow{I'}}\left(\sum_{i\in I'} c_i\alphalow{i}(\varepsilon_{i})\right),~
				\tilde{\varepsilon}\textstyle=\prod_{j \in \nbr_{\mathcal{I}}(I')} \set{\varepsilon_j}.
			\end{align*}
		 %where $\alphalow{I'}$ is witnessing $\Lypv{I'}$ constructed in \Lem~\ref{thm:compound lyapunov func}.% and $\nbr_{\mathcal{I}}(I') = \set{j\in I \mid \exists i\in I'.j\in\nbr_{\mathcal{I}}(i)\cap I\setminus I' }$.
			%\KM{$\alphalow{}$ is related to the overall Lyapunov function $V$ and comes from the proof of Lemma~\ref{thm:compound lyapunov func}.}
		%	if for all $i$, $\ContSysDT{\ContSysSymb{i}}$ and $\ContSysA{i}{\ContSysSymb{i}}$ are approximately bisimilar with the corresponding precision being $\varepsilon_{i}$ \st
	\end{corollary}

%	\begin{example}
%		Consider the systems in \Ex~\ref{ex2} and their abstractions. It follows from Corollary~\ref{cor:comp-appx} that the composition $\llbracket\ContSysA{i}{\Sigma_i}\rrbracket_{i\in \set{1,2}}$, given by the product automaton of the symbolic abstractions $\ContSysA{1}{\Sigma_1}$ and $\ContSysA{2}{\Sigma_2}$, is disturbance bisimilar to the composed system $\llbracket\ContSysDT{\Sigma_i}\rrbracket_{i\in \set{1,2}}$ with parameters $\varepsilon = \norm{[1 \ 1.5]^T} = 1.8028$ and $\tilde{\varepsilon} = \varepsilon_3 = 1$.
%	\end{example}

Recall that in the special case $I'=I$ the composed system replaces the overall network without extra external disturbances. In this case it is easy to see that the relation in Corollary~\ref{cor:comp-appx} simplifies to a usual bisimulation relation.

\begin{corollary}
 Given the premises of Corollary~\ref{cor:comp-appx} and that $I'=I$, the relation $\eqclass{}$ in \eqref{equ:eqclass} is an $\varepsilon$-approximate bisimulation relation between $\llbracket\ContSysDT{\Sigma_i}\rrbracket_{i\in I}$ and $\llbracket\ContSysAi{\Sigma_i}\rrbracket_{i\in I}$.
\end{corollary}
\section{Decentralized Controllers}\label{sec:control}
Finally, we discuss how our compositional approach leads to a decentralized controller synthesis methodology.

Let  $\Sigma$ be a control system and recall that $\ContSysDT{\Sigma}$ and $\ContSysA{}{\Sigma{}}$ are its induced metric systems as defined in Def.~\ref{def:control-sys-disctime} and Def.~\ref{def:abstract system}, respectively, which are related via the disturbance bisimulation relation $\eqclass{}$ in \eqref{eqn:local simulation relation} under the given assumptions. For these systems we denote by $\Xi$, $\Xi_\tau$ and $\Xi_{\AllQnt{}}$ the sets containing all their state trajectories. 
By slightly abusing notation, we furthermore assume in this section that $\ContSysA{}{\Sigma{}}$ was constructed over compact sets $X'\subset X$ and $U'\subset U$ as discussed in Rem.~\ref{rem:CompactStateset} and is therefore a finite metric system.

There are various types of specifications that can be used for controller synthesis. We assume in this paper that the specification is given as a subset $\varphi\subseteq\Xi$ of desired \emph{continuous} trajectories only taking values in the compact subset $X'\subset X$ of the state space. Given this set, we will design a controller in three steps. First, the specification $\varphi$ is abstracted to its time-discrete and abstract counterparts $\varphi_\tau$ and $\varphi_\AllQnt{}$. Second, a control function $\hat{f}^{c}$ is synthesized for the abstract metric system $\ContSysA{}{\Sigma{}}$ w.r.t.\ the specification 
$\varphi_\AllQnt{}$ forming the abstract closed-loop metric system $\ContSysAc{}{\Sigma{}}$ whose trajectories are guaranteed to be contained in $\varphi_\AllQnt{}$. Third, the closed loop system $\ContSysAc{}{\Sigma{}}$ is composed with the time-discrete metric system $\ContSysDT{\Sigma}$ using the constructed disturbance bisimulation $\eqclass{}$ to obtain a time-discrete closed loop system. Due to the properties of $\eqclass{}$, all trajecotries generated by this closed loop are guaranteed to be contained in $\varphi_\tau$. 
In addition to that, we can extend this soundness result to all continuous trajectories generated by this closed loop and show that they are contained in the set $\varphi$.

% \KM{One proposal for notation: how about using ``cl'' (in place of ``c'') as superscripts for the symbols related to closed loops?}

\subsection{Abstracting the Specification}

% Let $\varphi\subseteq\Xi$ be given. Then we define $\varphi_\tau$ and $\varphi_\AllQnt{}$ as follows.

Intuitively, the outlined abstraction based controller synthesis only provides a controller for the \emph{continuous} control system $\Sigma$ w.r.t.\ the original specification $\varphi$ if all continuous trajectories fulfilling $\varphi_\tau$ at sampling instances also fulfill $\varphi$ in inter-sampling periods. We can make this underlying assumption explicit by assuming that the vector field $f$ in the control system $\Sigma$ is bounded.

\begin{assumption}\label{ass:fbounded}
Let $\Sigma$ be a control system with $\dISS$ Lyapunov function $V$ satisfying \eqref{eq:ISS_cond}. % \new{s.t. \eqref{equ:sensitivity bound} %holds for the constant $\chi^{(2)} \in \mathbb{R}_{>0}$.} Let
	We assume there exists a constant $\chi > 0$ \st
	\begin{align}
		\norm{f(x,u,w)} \leq \chi
	\end{align}
	holds for any $x\in X$, $u\in U$, and $w\in W$.
\end{assumption}

% \begin{remark}
It should be noted that \Ass~\ref{ass:fbounded} implies the restriction posed by \Ass~\ref{ass:psi_i} on $f$. In other words, $\psi$ in \Ass~\ref{ass:psi_i} can always be calculated from a given $\chi$ in \Ass~\ref{ass:fbounded} and the maximum derivative of $\sigma_d$ on the compact interval $[0,\norm{\tilde\varepsilon}]$.
% \end{remark}

% Let $\chi$ be the maximum norm of the vector field $f$ in the control system $\Sigma$.
Given \Ass~\ref{ass:fbounded} and a trajectory $\xi_\tau\in\Xi_\tau$ we define the set of trajectories $\chi$-\emph{close} to $\xi_\tau$ by
\begin{align}\label{equ:psi_envelope}
% [\xi_\tau]_{\chi}=\left\{\right.&\xi'\in\Xi \mid \forall k\in\mathbb{N}, t\in[0,\tau]~.~ \nonumber\\
% &\left.|| \xi'(k\tau+t)-\xi_\tau(k)||\leq\chi \mathsf{min}(t,\tau-t)\right\}.%\label{eq:abstract_spec}
[\xi_\tau]_{\chi}=\left\{\right.\xi'\in\Xi& \mid \forall t\in\real{},\,\,k = \lfloor t/\tau+1/2\rfloor~.~ \nonumber\\
&\left.|| \xi'(t)-\xi_\tau(k)||\leq\chi \mathsf{min}(t,\tau-t)\right\},%\label{eq:abstract_spec}
\end{align}
where $\lfloor \cdot \rfloor$ is the floor function (largest integer not greater than its argument).
%\AKS{I did not spent much time on the above definition and I'm totally not sure about it as $\chi$ is defined for $\nu$, maybe $\sigma_d$ needs to go into it somewhere...}
% Sadegh: The definition is good in my opinion.
Given a set of desired continuous trajectories $\varphi\subseteq\Xi$, we can define $\varphi_\tau$ s.t.\ it contains all time sampled versions $\xi_\tau$ of trajectories $\xi\in\varphi$ whose $\chi$-envelope $[\xi_\tau]_{\chi}$ is also in $\varphi$, i.e.,
\begin{equation}\label{equ:varphi_tau}
 \varphi_\tau=\set{\xi_\tau\in\Xi_\tau \mid [\xi_\tau]_\chi\subseteq\varphi}.
\end{equation}

Similarly, given a trajectory $\xi_{\AllQnt{}}\in\Xi_{\AllQnt{}}$ and a disturbance bisimulation relation $\eqclass{}$ between $\ContSysDT{\Sigma}$ and $\ContSysA{}{\Sigma{}}$, we define the set of trajectories $\varepsilon$-close to $\xi_\AllQnt{}$ w.r.t.\ $\eqclass{}$ as
\begin{align}\label{equ:eqclass_envelope}
 [\xi_\AllQnt{}]_{\eqclass{}}=\set{\xi_\tau\in\Xi_\tau \mid
 \forall k\in\mathbb{N}~.~ (\xi_\tau(k),\xi_\AllQnt{i}(k))\in\eqclass{}
 }
\end{align}
resulting in the following definition for $\varphi_\AllQnt{}$;
\begin{equation}\label{equ:varphi_AllQnt}
 \varphi_\AllQnt{}=
 \set{\xi_\AllQnt{}\in\Xi_\AllQnt{} \mid 
 [\xi_\AllQnt{}]_{\eqclass{}}\subseteq\varphi_\tau
 }.
\end{equation}

% \begin{remark}
% Depending on the control problem at hand, the specification of interest might be directly defined over the abstract state space $X_\AllQnt{}$, for example in terms of a linear temporal logic (LTL) formula. Such an LTL formula can be directly translated into a set of \enquote{desired} trajectories $\varphi_\AllQnt{}$. In this case, the above abstraction procedure to obtain $\varphi_\tau$ and $\varphi_{\AllQnt{}}$ from a given $\varphi$ must be reversed. I.e., in this case $\varphi_\tau$ is given by all trajectories over $X$ that are in the $\varepsilon$-envelope of $\varphi_\AllQnt{}$ and, analogously, $\varphi$ is obtained from $\varphi_\tau$ by including all continuous trajectories that are in the $\chi$-envelope of $\varphi_\tau$. 
% \end{remark}

% \SEZS{
\begin{remark}
	Depending on the control problem at hand, the specification of interest might not be directly given as a set of continuous trajectories $\varphi$.  A common choice is to assign atomic propositions to subsets of the state space $X$ and employ linear temporal logic (LTL) over these propositions to express the specification of interest (as used in the example of \Sec~\ref{sec:Example}). 
	In this case, it is possible to directly use the LTL formula for the abstract controller synthesis by 
	properly shrinking or enlarging the labeled subsets of the state space to account for the abstraction errors (see e.g.,~\cite{Girard2012_ReachabilitySafety} for safety and reachability specifications). This methodology is contained in our setup as the LTL specification along with the respective state subsets can be translated into sets of desired trajectories $\varphi$, $\varphi_\tau$ and $\varphi_\AllQnt{}$ having the relationships captured by \eqref{equ:psi_envelope}-\eqref{equ:varphi_AllQnt}.
	It should be noted that soundness of this particular instance of abstraction-based controller synthesis relies on the implicit assumption that continuous trajectories behave nicely between inter-sampling periods, which is what \Ass~\ref{ass:fbounded} explicitly ensures.
% 	In this case, the above abstraction procedure to obtain $\varphi_\tau$ and $\varphi_{\AllQnt{}}$ from a given $\varphi$ must be reversed. I.e., in this case $\varphi_\tau$ is given by all trajectories over $X$ that are in the $\varepsilon$-envelope of $\varphi_\AllQnt{}$ and, analogously, $\varphi$ is obtained from $\varphi_\tau$ by including all continuous trajectories that are in the $\chi$-envelope of $\varphi_\tau$. 
% 	construct the set $\varphi_\tau$ directly from the formula without finding the subset of trajectories $\varphi$ by 
% 	For this purpose, we shrink any set $X_a\subset X$ associated to each label $a$ to a set $\tilde{X}_a\subset X_a$ such that $\xi_\tau(k)\in \tilde{X}_a$ implies $\xi'(k\tau+t)\in X_a$ for all $t\in[0,\tau)$ (cf. Eqn.~\eqref{eq:abstract_spec}). We also add a new label to the sets $X_a\backslash\tilde X_a$. As a result, satisfaction of any LTL formula on the new label set by a trajectory $\xi_\tau(\cdot)$ results in satisfaction of the same formula over the original labels by $\xi'(\cdot)$.
\end{remark}
% }

\subsection{Abstract Controller Synthesis}

Given a finite state metric system, such as $\ContSysA{}{\Sigma{}}$, a controller is a function $\hat{f}^{c}$ which restricts the available inputs in every state of $\ContSysA{}{\Sigma{}}$ s.t. a given property, i.e. $\varphi_\AllQnt{}$, is satisfied.
As any finite state metric system can be equivalently interpreted as a finite automaton, such controllers can be synthesized by well established techniques from reactive synthesis \cite{EmersonJutla91,MPS95}, whenever $\varphi_\AllQnt{}$ is an $\omega$-regular language. 
Such synthesized controllers are known to use \emph{finite} strings of past states visited by $\ContSysA{}{\Sigma{}}$ to reason about currently available inputs. For the ease of presentation, we restrict our attention to such control functions $\hat{f}^{c}$ that base their decisions solely on the currently available state\footnote{Technically, this implies that we restrict our attention to specifications $\varphi_\AllQnt{}$ that can be translated into a finite state automaton over the state space $X_\AllQnt{}$, as for example required in GR(1)-synthesis \cite{GR1_2012}, which is commonly used to synthesize controllers for cyber-physical systems.}, i.e., $\hat{f}^{c}: X_\AllQnt{}\rightarrow \mathcal{U}_{\AllQnt{}}$. 
% \KM{Shouldn't the controller, being an automated mechanism of controlling a plant, be a scalar-valued map $\hat{f}^{c}: X_\AllQnt{}\rightarrow {\mathcal{U}_{\AllQnt{}}}$? It is the set of all controllers which give us the set of all control actions at any given state. Shouldn't the maximality be interpreted differently than supervisory control, since in supervisory control the controller disables controllable actions, while in in our caser the controller enables controllable actions?}
For such functions it can be readily seen, that the closed-loop composed of the metric system $\ContSysA{}{\Sigma{}}$ and the control function $\hat{f}^{c}$ is given by the finite state metric system\footnote{If the general case for $\hat{f}^{c}$ is considered, $\ContSysAc{}{\Sigma{}}$ is defined over the 
product of the state space $X_\AllQnt{}$ and the memory structure $S\subseteq(X_\AllQnt{})^*$ required to define $\hat{f}^{c}:S\rightarrow\mathcal{U}_{\AllQnt{}}$.} $\ContSysAc{}{\Sigma{}}$ which is equivalent to $\ContSysA{}{\Sigma{}}$ up to the transition function, given by
\begin{equation}\label{equ:deltachat}
 \hat{x}'\in\delta^{c}_{\AllQnt{}}(\hat{x},\mu,\nu) \Leftrightarrow  \hat{x}'\in\delta_{\AllQnt{}}(\hat{x},\mu,\nu) \wedge \mu= \hat{f}^{c}(\hat{x}).
\end{equation}

Given the soundness of reactive controller synthesis, we have the following guarantee on the behavior of $\ContSysAc{}{\Sigma{}}$.

\begin{proposition}\label{prop:controllersound}
 Let $\ContSysA{}{\Sigma{}}$ be the abstract state metric system constructed in \Def~\ref{def:abstract system} and $\varphi_\AllQnt{}\subseteq\Xi_\AllQnt{}$ be a specification. If $\hat{f}^{c}: X_\AllQnt{}\rightarrow \mathcal{U}_{\AllQnt{}}$ is a controller for $\ContSysA{}{\Sigma{}}$ w.r.t.\ $\varphi_\AllQnt{}$ then 
 \begin{equation}\label{equ:Xicinvar}
  \Xi_\AllQnt{}^{c}\subseteq\varphi_\AllQnt{},
 \end{equation}
 where
 \begin{equation*}
  \Xi_\AllQnt{}^{c}=\set{\xi\in\Xi_\AllQnt{} \mid \forall k~.~\exists \mu,\nu ~.~ \xi(k+1)\in\delta^{c}_\AllQnt{}(\xi(k),\mu,\nu)}
 \end{equation*}
 with $\delta^{c}_\AllQnt{}$ as in \eqref{equ:deltachat}.
\end{proposition}

\subsection{Controller Refinement}  
% When using the control function $\hat{f}^{c}$ for the control of the time-sampled transition system $\ContSysDT{\Sigma{}}$, it should be observed, that we 
Unfortunately, we cannot simply refine $\hat{f}^{c}:X_\AllQnt{}\rightarrow \mathcal{U}_{\AllQnt{}}$ to a control function $f^{c}:X\rightarrow \mathcal{U}_{\tau}$ to be applied to the time-sampled transition system $\ContSysDT{\Sigma{}}$. This is due to the way disturbance bisimulations are set up. Given a disturbance bisimulation relation $\eqclass{}$ between $\ContSysDT{\Sigma{}}$ and $\ContSysA{}{\Sigma{}}$, every continuous state $x\in X$ might be related to various abstract states $\hat{x}$. Therefore one needs to run the controlled abstract model $\ContSysAc{}{\Sigma{}}$ alongside with $\ContSysDT{\Sigma{}}$ to apply the right control action\footnote{This can be avoided when the disturbance bisimulation relation defined in \Def~\ref{def:DisturbanceBisimulation} would be strengthened to a feedback-refinement relation (see \cite{ReissigWeberRungger_2017_FRR}).} in the current state $x$. 
% \KM{Should we mention that this not the case in our setting, and the abstract closed loop can be run independently?}
This is modeled by the following product construction of two metric systems adapted from \cite[\Def 11.9]{TabuadaBook}.

Given $\ContSysAc{}{\Sigma{}}$ as in \eqref{equ:deltachat} and $\ContSysDT{\Sigma{}}$ as in \eqref{def:control-sys-disctime} we define their composition  w.r.t.\ the disturbance bisimulation $\eqclass{}$ in \eqref{eqn:local simulation relation} as  the metric system $\ContSysDT{\Sigma{}}\times_{\eqclass{}} \ContSysAc{}{\Sigma{}} = (X', U_\AllQnt, \mathcal{U}_{\AllQnt{}}, W_\tau\times \widetilde{W}, \mathcal{W}_{\tau}\times\mathcal{W}_\AllQnt{}, \delta'_{\tau})$, s.t.\
	\begin{align}\label{equ:constructTimes}
		&X' = \set{ (x,\hat{x})\in X_\tau\times X_\AllQnt{} \mid (x,\hat{x})\in \eqclass{}}\quad\text{and}\\
		&\delta'_\tau((x,\hat{x}),\mu,(\nu,\hat{\nu}))=\left\{(x',\hat{x}')\in \eqclass{} \mid 
		\begin{matrix}
		\phantom{\wedge}~\hat{x}'\in\delta^{c}_{\AllQnt{}}(\hat{x},\mu,\hat{\nu})\\
		\wedge~x'\in\delta_\tau(\hat{x},\mu,\nu)%\\
% 		\wedge (\hat{x}',x')\in\eqclass{}
		\end{matrix}\right\}.\notag
	\end{align} 

Using this product as our closed loop system, we get the following soundness result.

\begin{theorem}\label{thm:monolithicControllerSynth}
Given the preliminaries of \Thm~\ref{thm:sim-approx-single}, \Ass~\ref{ass:fbounded} and $\ContSysAc{}{\Sigma{}}$ as in \eqref{equ:deltachat}, 
% the composed metric system $\ContSysDT{\Sigma{}}\times_{\eqclass{}} \ContSysAc{}{\Sigma{}}$ as defined in \eqref{} fullfills
it holds that
\begin{equation}
  \Xi_\tau^{c}\subseteq\varphi_\tau \quad\text{and}\quad \Xi^{c}\subseteq\varphi
 \end{equation}
 where
 \begin{align}
  \Xi_\tau^{c}=&\left\{\xi\in\Xi_\tau\right. \mid \exists \hat{\xi}\in\Xi_\AllQnt{}~.~\forall k~.~\exists \mu,\nu,\hat{\nu} ~.~\label{equ:Xitauc}\\
  &\left.(\xi(k+1),\hat{\xi}(k+1))\in\delta'_\tau((\xi(k),\hat{\xi}(k),\mu,(\nu,\hat{\nu}))\right\}\notag\\
%  \end{align*}
% %  and 
% \begin{align*}
 \Xi^{c}=&\set{\xi\in\Xi \mid \exists \xi_\tau\in\Xi_\tau^{c}~.~\forall k\in\mathbb{N}~.~\xi(k\tau)=\xi_\tau(k)}.\label{equ:XiAllc}
\end{align}
\end{theorem}

\begin{proof}
We prove both claims separately. \\
$\Xi_\tau^{c}\subseteq\varphi_\tau$: Pick any trajectory $\xi_\tau\in\Xi_\tau^{c}$. Then it follows from \eqref{equ:Xitauc} that there exists a trajectory $\xi_\AllQnt{}\in\Xi_\AllQnt{}$ s.t. for all  $k$ there exists $\mu$, $\nu$ and $\hat{\nu}$ s.t. it holds that $(\xi_\tau(k+1),\xi_\AllQnt{}(k+1))\in\delta'_\tau((\xi_\tau(k),\xi_\AllQnt{}(k)),\mu,(\nu,\hat{\nu}))$. Using \eqref{equ:constructTimes} this has two consequences; for all $k$ we have 
\begin{inparaenum}[(i)]
 \item $(\xi_\tau(k+1),\xi_\AllQnt{}(k+1))\in \eqclass{}$, and 
 \item $\xi_\AllQnt{}(k+1)\in\delta^{c}_{\AllQnt{}}(\xi_\AllQnt{i}(k),\mu,\hat{\nu})$.
\end{inparaenum}
Now it can be observed that (i) implies $\xi_\tau\in[\xi_\AllQnt{}]_{\eqclass{}}$ (from \eqref{equ:eqclass_envelope}) and (ii) implies $\xi_\AllQnt{}\in\varphi_\AllQnt{}$ (from \Prop~\ref{prop:controllersound}) and therefore $\xi_\tau\in\varphi_\tau$ (from \eqref{equ:varphi_AllQnt}).\\
% \SEZS{
$\Xi^{c}\subseteq\varphi$: For the second claim pick $\xi\in\Xi^{c}$ and observe that this implies the existence of $\xi_\tau\in\Xi_\tau^{c}$ \st for all $k$ holds that $\xi(k\tau)=\xi_\tau(k)$ (from \eqref{equ:XiAllc}).
On the one hand, from the first claim $\Xi_\tau^{c}\subseteq\varphi_\tau$ proved above we get $\xi_\tau\in\varphi_\tau$, which then implies $[\xi_\tau]_\chi\subseteq\varphi$ (from \eqref{equ:varphi_tau}).
On the other hand, the rate of changes in $\xi$ is bounded by $\chi$ and $\xi(k\tau)=\xi_\tau(k)$ for all $k$, which result in $\xi(t)$ satisfying the inequality in \eqref{equ:psi_envelope}.
Therefore $\xi\in [\xi_\tau]_\chi$.
These two together imply $\xi\in\varphi$.
% }
\end{proof}

\subsection{Compositional Soundness}
Analogously to the generalization of \Thm~\ref{thm:sim-approx-single} to networks of metric systems in \Cor~\ref{cor:sim-appx}, we can generalize \Thm~\ref{thm:monolithicControllerSynth} to the following corollary derived from  \Cor~\ref{cor:sim-appx}.

\begin{corollary}
 Given the preliminaries of \Cor~\ref{cor:sim-appx}, let $\Xi_i$, $\Xi_{\tau,i}$ and $\Xi_{\AllQnt{i}}$ be the sets of trajectories of $\Sigma_i$, $\ContSysDT{\Sigma_{i}}$ and $\ContSysAc{i}{\Sigma_{i}}$, respectively. Furthermore, let $\varphi_i\subseteq\Xi_i$ be given for all $i\in\mathcal{N}$, and $\varphi_{\tau,i}\subseteq\Xi_{\tau,i}$ and $\varphi_{\AllQnt{i}}\subseteq\Xi_{\AllQnt{i}}$ be the abstract specifications of $\varphi_i$ as constructed in \eqref{equ:varphi_tau} and \eqref{equ:varphi_AllQnt}, respectively. Finally, let $\ContSysAc{i}{\Sigma{}}$ be the controlled abstract system fulfilling \eqref{equ:Xicinvar} and let \Ass~\ref{ass:fbounded} hold for every $i\in\mathcal{N}$. Then we can define $\Xi_{\tau,i}^{c}$ and $\Xi^{c}_i$ analogously to \eqref{equ:Xitauc} and \eqref{equ:XiAllc} and it holds that
 \begin{equation}
  \Xi_{\tau,i}^{c}\subseteq\varphi_{\tau,i} \quad\text{and}\quad \Xi^{c}_i\subseteq\varphi_i.
 \end{equation}
\end{corollary}

\section{An example}\label{sec:Example}
% First we illustrate the abstraction and synthesis procedure developed in this paper by the network of systems depiced in \Fig.~\ref{fig:ex1}. Second we show the effectiveness using a very large system. 
Consider the following interconnected linear time invariant systems
\begin{align*}
	&\Sigma_{1,i}: \ \begin{bmatrix}
						\dot{x}_{1,i}\\
						\dot{x}_{2,i}
					\end{bmatrix} = \begin{bmatrix}
										-1 & 1\\
										-1 & -1
									\end{bmatrix}\begin{bmatrix}
										x_{1,i}\\
										x_{2,i}
									\end{bmatrix}
									+
									\begin{bmatrix}
										0 & 0 & 0\\
										1 & 0.1 & 0.1
									\end{bmatrix}\begin{bmatrix}
										u_1\\
										x_{2,i-1}\\
										x_{4,i}
									\end{bmatrix}\\
	&\Sigma_{2,i}: \ \begin{bmatrix}
						\dot{x}_{3,i}\\
						\dot{x}_{4,i}
					\end{bmatrix} = \begin{bmatrix}
										-1 & 1\\
										-1 & -1
									\end{bmatrix}\begin{bmatrix}
										x_{3,i}\\
										x_{4,i}
									\end{bmatrix}
									+
									\begin{bmatrix}
										0 & 0\\
										1 & 0.3
									\end{bmatrix}\begin{bmatrix}
										u_2\\
										x_{2,i}
									\end{bmatrix}
\end{align*}
where $i\in [1,N]$ and $x_{2,i-1}=0$ for $i=1$. The system is made up of $N$ identical smaller networks in cascade, where each smaller network consists of a pair of two linear time-invariant control systems connected in feedback. One such instance of the network consisting of $N=3$ such pairs is depicted in \Fig~\eqref{fig:ex1}.
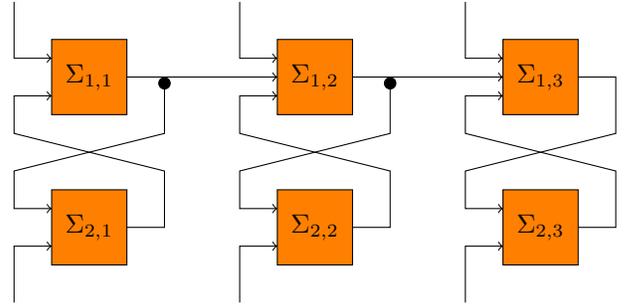
\begin{figure}[h]
	\centering
	\begin{tikzpicture}
		%\draw[step=1.0,black,thin] (-1,-1) grid (8,4);
		\draw[fill=orange]	(0,0)	rectangle	(1,1)	node[pos=0.5]		{$\Sigma_{2,1}$};
		\draw[fill=orange]	(0,2)	rectangle	(1,3)	node[pos=0.5]		{$\Sigma_{1,1}$};
		\draw[fill=orange]	(3,0)	rectangle	(4,1)	node[pos=0.5]		{$\Sigma_{2,2}$};
		\draw[fill=orange]	(3,2)	rectangle	(4,3)	node[pos=0.5]		{$\Sigma_{1,2}$};
		\draw[fill=orange]	(6,0)	rectangle	(7,1)	node[pos=0.5]		{$\Sigma_{2,3}$};
		\draw[fill=orange]	(6,2)	rectangle	(7,3)	node[pos=0.5]		{$\Sigma_{1,3}$};
		
		%\draw[->]	(1,0.5)		--		(3,0.5);
		\draw[->]			(1,2.5)		--		(3,2.5);
		%\draw[->]			(4,0.5)		--		(6,0.5);
		\draw[->]			(4,2.5)		--		(6,2.5);
		
		\draw[->]	(1,0.5)	--	(1.5,0.5)	--	(1.5,1.25)	--	(-0.5,1.75)	--	(-0.5,2.25)	--	(0,2.25);
		\draw[->]	(4,0.5)	--	(4.5,0.5)	--	(4.5,1.25)	--	(2.5,1.75)	--	(2.5,2.25)	--	(3,2.25);
		\draw[->]	(7,0.5)	--	(7.5,0.5)	--	(7.5,1.25)	--	(5.5,1.75)	--	(5.5,2.25)	--	(6,2.25);
		
		\draw[*->]	(1.5,2.5)	--	(1.5,1.75)	--	(-0.5,1.25)	--	(-0.5,0.75)	--	(0,0.75);
		\draw[*->]	(4.5,2.5)	--	(4.5,1.75)	--	(2.5,1.25)	--	(2.5,0.75)	--	(3,0.75);
		\draw[->]	(7,2.5)	--	(7.5,2.5)	--	(7.5,1.75)	--	(5.5,1.25)	--	(5.5,0.75)	--	(6,0.75);
		
		\draw[->]	(-0.5,-0.5)	--	(-0.5,0.25)	--	(0,0.25);
		\draw[->]	(2.5,-0.5)	--	(2.5,0.25)	--	(3,0.25);
		\draw[->]	(5.5,-0.5)	--	(5.5,0.25)	--	(6,0.25);
		
		\draw[->]	(-0.5,3.5)	--	(-0.5,2.75)	--	(0,2.75);
		\draw[->]	(2.5,3.5)	--	(2.5,2.75)	--	(3,2.75);
		\draw[->]	(5.5,3.5)	--	(5.5,2.75)	--	(6,2.75);
	\end{tikzpicture}
	\caption{Fig: Network of control systems for $N=3$.}
	\label{fig:ex1}
\end{figure}
% Note that the choice of the identical structure is solely because of the ease of presentation: we are going to build a large chain of cascaded systems and going to use copies of the same controllers.% This similarity allows us to reuse the same controller for all the copies of the system in the network.

Let the compact state and input spaces of $\Sigma_{1,i}$ and $\Sigma_{2,i}$ considered for the construction of abstract metric systems be given by $X'_{1,i} = [-3.2,3.2]\times [-3.2,3.2]$, $X'_{2,i} = [-4.2,4.2]\times [-4.2,4.2]$, $U'_{1,i}=[-5,5]$ and $U'_{2,i}=[-7,7]$ respectively. Each of the systems in the network has reachability and safety specifications given in LTL:%, which are given in LTL as follows:%with the target sets being given by the following ellipsoids
\begin{align}\label{equ:exampleLTL}
	\varphi_{1,i}= \lozenge R_{1,i}\wedge\square B_{1,i} \qquad \qquad	\varphi_{2,i}= \lozenge R_{2,i} \wedge \square B_{2,i}
\end{align}
where $\lozenge$ means ``eventually'' (reachability) and $\square$ means ``always'' (safety), and $R_{1,i}$, $R_{2,i}$ are ellipsoidal sets of target states: 
$R_{1,i}=\set{[x_{1,i} \ x_{2,i}]^T\in X'_{1,i}\mid (x_{1,i}-1.5)^2+x_{2,i}^2 \leq 0.94}$, $R_{2,i}=\set{[x_{3,i} \ x_{4,i}]^T\in X'_{2,i}\mid (x_{3,i}+1.5)^2+x_{4,i}^2 \leq 0.7}$, and $B_1$, $B_2$ are sets of safe states: $B_1 = X'_{1,i}\setminus [-1,0.5]\times[-1.5,1.5]$ and $B_2 = X'_{2,i}\setminus [-1,0.5]\times[-1.5,1.5]$ (i.e. the rectangle $[-1,0.5]\times[-1.5,1.5]$ is an obstacle for both $\varphi_{1,i}$ and $\varphi_{2,i}$).% respectively.

For this given set of systems and specifications, we wish to synthesize decentralized controllers s.t. $\varphi_{1,i}$ and $\varphi_{2,i}$ are satisfied by each individual $i$-th closed loop. Actually, by taking advantage of the similarity of the specifications and the dynamics, we just need to synthesize two closed loops and deploy identical copies of them in each subsystem. 

As the prerequisite of controller synthesis, we first point out that both $\Sigma_{1,i}$ and $\Sigma_{2,i}$ admit $\dISS$ Lyapunov functions that are presented together with their associated parameters in Table~\ref{table:lyapunov}. %The associated parameters $\psi_{k,i}$-s (as defined in \Ass~\ref{ass:psi_i}) for $k\in\set{1,2}$ are given by: $\psi_{1,i} = 4.7405$ and $\psi_{2,i} = 3.3541$ for all $i\in N$.

\begin{table}[h]
	\centering
		\caption{$\dISS$ Lyapunov functions and the corresponding parameters.}
	\label{table:lyapunov}
	\begin{tabular}{l | c c}
% 		\hline\\
		& $\Sigma_{1,i}$	& $\Sigma_{2,i}$\\
		\hline\\
		$V$	&	$5x_{1,i}^2+5x_{2,i}^2$	&	$5x_{3,i}^2+5x_{4,i}^2$\\[0.2cm]
		$\alphalow{}$ &	$2.2361$	&	$2.2361$\\
		$\alphahigh{}$	&	$2.2361$	&	$2.2361$\\
		$\gamma$	&	$2.2361$	&	$2.2361$\\
		$\lambda$	&	$1$	&	$1$\\
		$\sigmau{}$	&	$2.2361$	&	$2.2361$\\
		$\sigmad{}$	&	$0.3162$	&	$0.6708$\\[0.2cm]
		$\psi$		&	$4.7405$	&	$3.3541$
% 		\hline
	\end{tabular}
\end{table}

Given this setup we discuss two different cases.%show simulation results for $N=3$ and $N=100$ in \eqref{fig:closed loop 3sys} and \eqref{fig:closed loop allinone} as discussed below.

\paragraph{N=3}
 Using the parameters given in Table~\ref{table:lyapunov} it can be verified that for $A$, $B$ defined as in \Thm~\ref{thm:cond:sim-appx}, $\lambda_{max}(A^{-1}B) = 0.4606 < 1$. Then by Remark~\ref{rem:small gain}, we have that \eqref{equ:small gain kind} holds for $N=3$.
Now we fix the abstraction parameters as follows: $\tau = 0.1$, $\omega_1=\omega_2=0.1$ and $\varepsilon_1 = \varepsilon_2 = 0.7$. Using this set of parameters and the Lyapunov functions in Table~\ref{table:lyapunov}, \eqref{cor:sim-appx} evaluates to $0<\eta_{1,i}<0.0236$ and $0<\eta_{2,i}<0.0228$. Then by \Cor~\ref{cor:sim-appx}, we have that the finite state abstractions $\ContSysAi{\Sigma_{j,i}}$ are disturbance bisimilar with parameters $(\varepsilon_j,\tilde{\varepsilon}_j)$ to the sampled time systems $\ContSysDT{\Sigma_{j,i}}$ for $j\in \set{1,2}$.

\begin{figure*}[t]
	\centering
	\includegraphics[scale=1]{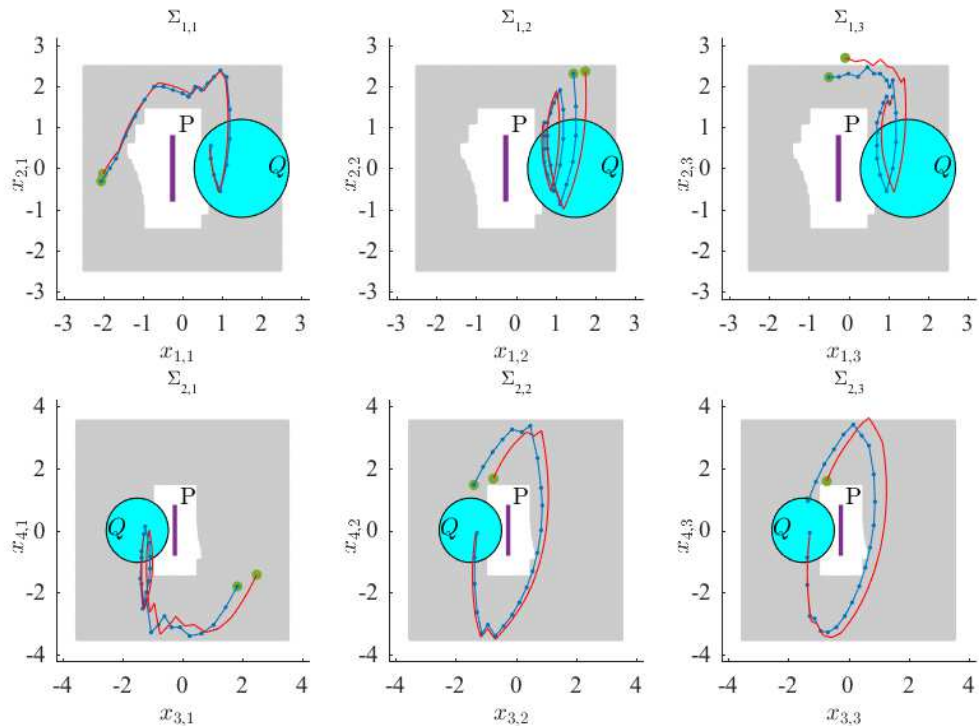}
	\caption{Simultaneous evolution of the state trajectories in the network of closed loop systems for $N=3$, with arbitrarily chosen initial states within the domain of the controllers. For each subplot, the gray region is the domain of the abstract controller, the purple rectangle ($P$) in the middle is the obstacle, and the cyan circle ($Q$) is the target of the reachability objective. The red and blue lines are the continuous and abstract trajectories respectively, which start from the green dots.}
	\label{fig:closed loop 3sys}
\end{figure*}

In this example, because of the similarity of the subsystems and their specifications, we only need to solve two synthesis problems; we synthesize two control functions $\hat{f}^c_1$ and $\hat{f}^c_2$ for $\ContSysAi{\Sigma_{1,i}}$ and $\ContSysAi{\Sigma_{2,i}}$ w.r.t.\ the specifications $\varphi_{\AllQnt{i},1,i}$ and $\varphi_{\AllQnt{i},2,i}$, respectively, for some $i\in [1,N]$. In this particular case, $\varphi_{\AllQnt{},j,i}$ is the LTL specification in \eqref{equ:exampleLTL} over the $\varepsilon_i$-deflation of the target and save sets. 
% Let $C_1$ and $C_2$ be the controllers s.t. $C_1\times \ContSysAi{\Sigma_{1,i}} \cong_{\varepsilon_1} \varphi_{1,i}$ and $C_2\times \ContSysAi{\Sigma_{2,i}} \cong_{\varepsilon_2} \varphi_{2,i}$.
Refining these abstract controllers as discussed in \Sec~\ref{sec:control} results in a network of closed loop systems whose simultaneously generated trajectories are depicted in \Fig~\ref{fig:closed loop 3sys}. The simulation was stopped after  each of the systems has fulfilled its reachability objective at least once. \Fig~\ref{fig:closed loop 3sys} shows that all local closed loops robustly and independently satisfy their objectives.

\paragraph{N=100}
Now we increase the size of the system to $N=100$, with a total number of $400$ state variables. To our best knowledge, no existing tool for monolithic synthesis scales to such a large system. However, our method scales perfectly as controller synthesis only needs to be performed for systems with two state variables as discussed before. The resultant continuous trajectories of the network of closed loop systems are depicted in \Fig~\ref{fig:closed loop allinone}. For clarity of presentation every trajectory was stopped when it first met its reachability objective. It is observed that each of the subsystems fulfills it's specification.

We want to point out that $N$ could have been increased to any arbitrarily large value without affecting the sound behavior of the local controllers for each subsystem. The reason is that the abstraction error of each subsystem in the network is immune to the abstraction error of non-neighboring subsystems. This is easy to verify from \Inq~\eqref{eq:condition on StateQnt i}, where we use only the \emph{upper bounds} (i.e. the most pessimistic bounds) on the abstraction errors of the neighbors in $\tilde{\varepsilon}_i$. Since the abstraction error of each subsystem does not depend on the non-neighboring subsystems, and moreover the number of neighbors of all but one subsystems in the network remain the same when we increase $N$, no matter what value $N$ might take soundness is guaranteed.%, which in turn remains unaffected.
\begin{figure}%[h]
	\centering
	\includegraphics[scale=1]{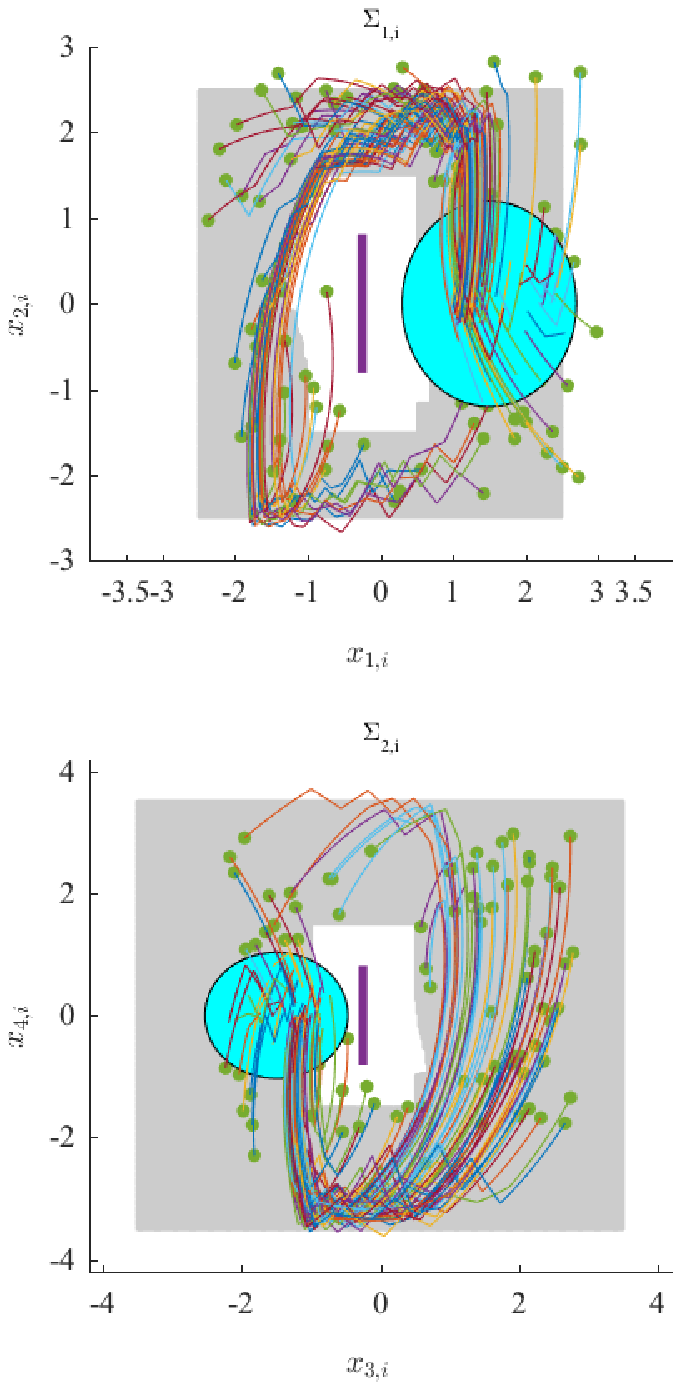}
	\caption{Simultaneous evolution of the systems for $N=100$ with arbitrary initial points within the domain of the controllers. For each subplot, the state space has the same representation same as in \Fig~\ref{fig:closed loop 3sys} (annotations are omitted). The lines represent the continuous trajectories of various systems which start from the green dots.}
	\label{fig:closed loop allinone}
\end{figure}

\section{Conclusion}

In this paper we introduced disturbance bisimulation as an equivalence relation between two metric systems having the same metric on their state spaces, and showed that disturbance bisimulation is closed under system composition. We extended disturbance bisimulation to two different abstractions of nonlinear dynamic systems by suitably abstracting the time, input-space and state-space. Finally we show how exploiting the closure under composition property, one can use disturbance bisimilar abstractions for decentralized controller synthesis with omega-regular control objectives. We demonstrate the effectiveness of our theory by an example.

%\input{AA}

%\input{linear}

%\input{implementation}

% \input{discussion}
% \newpage
\bibliographystyle{IEEEtran}
\bibliography{reportbib}

\end{document}